\documentclass[journal]{IEEEtran}

\usepackage{amsmath} 
\usepackage{amssymb}  
\usepackage{graphicx}
\usepackage{amsmath}
\usepackage{epstopdf}
\usepackage{mathtools}
\usepackage{dsfont}
\usepackage{url}
\usepackage[ruled,vlined,linesnumbered]{algorithm2e}
\usepackage{color}
\usepackage{booktabs} 
\usepackage{multirow}   
\usepackage{balance}    
\usepackage{soul}   
\usepackage{hyperref}   
\usepackage{empheq}

\usepackage[noadjust]{cite} 

\makeatletter
\newcommand{\removelatexerror}{\let\@latex@error\@gobble}
\makeatother

\DeclareMathOperator{\Succ}{Succ}
\newcommand{\beq}{\begin{equation}}
\newcommand{\eeq}{\end{equation}}

\usepackage{bm}

\newcounter{algorithmctr}[section]
\renewcommand{\thealgorithmctr}{\thesection.\arabic{algorithmctr}}
   {\refstepcounter{algorithmctr}\begin{list}{}{%
       \setlength{\rightmargin}{0\linewidth}%
       \setlength{\leftmargin}{.05\linewidth}
        \setlength{\itemsep}{1pt}
  \setlength{\parskip}{0pt}
  \setlength{\parsep}{0pt}}%
       \rmfamily\small
       \item[]{\setlength{\parskip}{0ex}\hrulefill\par%
        \nopagebreak{\bfseries\textsf{Algorithm \thealgorithmctr~}}}}%
   {{\setlength{\parskip}{-1ex}\nopagebreak\par\hrulefill} \end{list}}
\IEEEoverridecommandlockouts
\newtheorem{assumption}{Assumption}
\newtheorem{theorem}{Theorem}
\newtheorem{proposition}{Proposition}

\newtheorem{remark}{Remark}

\newtheorem{definition}{Definition}

\def\sq{\mathbin{{\strut\rule{1.25ex}{1.25ex}}}}
\newenvironment{proof}{{\textit{Proof:}}}{\hfill$\sq$}

\def\mathcolor#1#{\@mathcolor{#1}}
\def\@mathcolor#1#2#3{%
  \protect\leavevmode
  \begingroup
    \color#1{#2}#3%
  \endgroup
}

\title{\LARGE \bf Robust Learning Model Predictive Control for \\ Linear Systems Performing Iterative Tasks}

\author{Ugo Rosolia, Xiaojing Zhang, and Francesco Borrelli
\thanks{U.\ Rosolia, X.\ Zhang and F.\ Borrelli are with the Department of Mechanical Engineering, University of California at Berkeley ,
        Berkeley, CA 94701, USA
        {\tt\small\{ugo.rosolia, xiaojing.zhang,\ fborrelli\}{@}berkeley.edu}}%
}

\begin{document}

\maketitle
\thispagestyle{empty}
\pagestyle{empty}

\begin{abstract}
A robust Learning Model Predictive Controller (LMPC) for uncertain systems performing iterative tasks is presented. 
At each iteration of the control task the closed-loop state, input, and cost are stored and used in the controller design. 
This paper first illustrates how to construct robust control invariant sets and safe control policies exploiting historical data.
Then, we propose an iterative LMPC design procedure, where data generated by a robust controller at iteration $j$ are used to design a robust LMPC at the next iteration $j+1$. 
We show that this procedure allows us to iteratively enlarge the domain of the control policy and it guarantees recursive constraints satisfaction, input to state stability, and performance bounds for the certainty equivalent closed-loop system.
The use of different feedback policies along the horizon is  the key element of the proposed design.
The effectiveness of the proposed control scheme is illustrated on a linear system subject to bounded additive disturbances. 
\end{abstract}

\section{Introduction}
Exploiting historical data  to iteratively improve the performance of Model Predictive Controllers (MPC) has been an active theme of research over the past few decades \cite{rosolia2018data, hewing2019learning, aswani2013provably,kocijan2004gaussian,koller2018learning,berkenkamp2016safe,hewing2018cautious,hewing2019cautious,terzi2018learning,bacic2003general,ostafew2014learning,ostafew2016robust,brunner2013stabilizing,rosolia2019learning,lee1997model,lee2000convergence,lee2000model,lu2019robust,lorenzen2019robust,lorenzen2017adaptive,tanaskovic2013adaptive,tanaskovic2014adaptive,bujarCDC18,bujarbaruahAdapFIR,bujarArxivAdap}. The key idea is to use the stored state, input, and cost data to compute at least one of the following control design elements: $\emph{i})$ a \textit{model} to  predict the system trajectory for a given initial state and input sequence, $\emph{ii})$ a \textit{safe set} of states from which the control task can be completed using a known safe policy and $\emph{iii})$ a \textit{value function}, which for a given safe policy, maps each state of the safe set to the closed-loop cost of completing the task. 

Policy evaluation strategies used to estimate value functions from historical data are studied in Approximate Dynamic Programming (ADP) and Reinforcement Learning (RL)~\cite{bertsekas1996neuro,bertsekas2019feature,recht2018tour}. For instance, direct strategies compute the estimate value function which best fits the realized closed-loop cost over the stored states. On the other hand, in indirect strategies the estimate value function is computed by iteratively minimizing the temporal difference~\cite{sutton1988learning,bradtke1996linear}. A survey on policy evaluation strategies goes beyond the scope of this paper, we refer the reader to~\cite{bertsekas2019feature,bertsekas1996neuro} for a comprehensive review on this topic.


The integration of MPC with system identification  strategies used to estimate the prediction model has been extensively studied in literature~\cite{aswani2013provably,kocijan2004gaussian,berkenkamp2016safe,koller2018learning,hewing2018cautious,hewing2019cautious,rosolia2019learning, terzi2018learning,bacic2003general,brunner2013stabilizing,ostafew2014learning,ostafew2016robust}. In adaptive MPC strategies~\cite{lu2019robust,lorenzen2019robust,lorenzen2017adaptive,tanaskovic2013adaptive,tanaskovic2014adaptive,bujarCDC18,bujarbaruahAdapFIR,bujarArxivAdap}, set-membership approaches are used to identify the set of possible parameters and/or the domain of the uncertainty which characterize the system's model. Afterwards, robust MPC strategies for additive~\cite{bemporad1999robust} or parametric~\cite{fleming2014robust,evans2012robust} uncertainty are used to guarantee recursive constraint satisfaction. Another strategy to identify the system dynamics is to fit a Gaussian Process (GP) to experimental data~\cite{koller2018learning,hewing2018cautious,kocijan2004gaussian,berkenkamp2016safe, hewing2019cautious}. GP can be used to identify a nominal model and confidence bounds, which may be used to tighten the constraint set over the planning horizon. This strategy provides high-probability safety guarantees~\cite{koller2018learning, hewing2019cautious}. The effectiveness of GP-based strategies on experimental platforms has been shown in~\cite{hewing2018cautious, berkenkamp2016safe}, where an MPC is used to race a 1/43-scale vehicle and to safely fly a drone. Regression strategies may also be used to identify the system model~\cite{rosolia2019learning, terzi2018learning}. For instance, the authors in~\cite{terzi2018learning} used a linear regression strategy to identify both a nominal model and the disturbance domain used for robust MPC design. In~\cite{rosolia2019learning}, we used local linear regression to identify the system model used by the controller, which was able to drive a 1/10-scale race car at the limit of handling.

Model-based and data-based approaches for computing safe sets have also been proposed in literature \cite{fisac2018general,kaynama2011continual,lygeros1999controllers, liniger2017real, wabersich2018linear, bacic2003general,brunner2013stabilizing,blanchini2005relatively, LMPClinear, LMPC}. In reachability-based strategies safe sets are computed solving a two players game between the controller and the disturbance \cite{fisac2018general,kaynama2011continual,lygeros1999controllers}.
Furthermore, these strategies provide a control policy, which can be used to guarantee safety by robustly constraining the evolution of the system within the safe set \cite{fisac2018general}. 
Also viability theory may be used to compute safe sets \cite{liniger2017real}. The authors in \cite{liniger2017real} showed how to compute an inner approximation of the viability kernel and demonstrated the effectiveness on a RC-car set-up. In \cite{wabersich2018linear} the authors showed how to compute safe sets for uncertain systems exploiting data from a robust controller, afterwards they used the safe set in a linear model predictive safety certification framework. In~\cite{bacic2003general,brunner2013stabilizing} the authors computed safe sets and the associated control policy combining stored trajectories with polyhedral and ellipsoidal invariant sets. Another approach is proposed in~\cite{blanchini2005relatively} where the stored trajectories are mirrored to construct invariant sets. Finally, in \cite{LMPClinear, LMPC} we have shown how data from a deterministic system can be trivially used to compute safe sets. However, these strategies cannot be used to compute safe sets for uncertain systems. 

In this paper we present an iterative \emph{robust} Learning Model Predictive Control (LMPC) design procedure for uncertain systems. We consider control tasks where the goal is to minimize a given cost function while satisfying state and input constraints. An iteration is a finite time execution of a control task from a given initial condition. We assume that at each iteration  the constraint set and cost function are unchanged, but the initial condition and duration can be selected by the control designer.
At each iteration we exploit historical data and the LMPC policy from previous iterations to construct robust safe sets and approximations to the value function. Furthermore, we show how to compute a safe control policy which is defined over the robust safe set and it is used in the LMPC design. In particular, the controller plans the system trajectory using either a safe policy or a disturbance feedback policy.
We show that the proposed strategy guarantees that: \emph{i}) state and input constraints are recursively satisfied, \emph{ii}) the closed-loop system is Input-to-State Stable (ISS) and \emph{iii}) the performance of the certainty equivalent system is bounded by a function $Q^j(\cdot)$ which is non-increasing with the iteration index (i.e. $Q^{j+1}(\cdot) \leq Q^j(\cdot)$). 
Finally, we show that the proposed iterative design procedure can be used to collect data on progressively larger regions of the state space. Thus, in applications where a conservative robust controller is available, our iterative algorithm may be used to improve the closed-loop performance and to enlarge the region of attraction.

Compared with standard robust MPC strategies~\cite{bemporad1999robust, mayne2005robust, chisci2001systems, bemporad2003min, chen1998robust}, where the terminal cost and constraint set used in the design are fixed, the proposed strategy iteratively updates these components used to synthesize the control policy. This strategy, which is tailored to iterative tasks, allows us to iteratively enlarge the region of attraction of the controller and to improve the performance of the certainty equivalent closed-loop system. 
Compared to the robust certainty equivalent LMPC in~\cite{rosolia2017robust}, where we constructed a control invariant set and a control Lyapunov function for the certainty equivalent system, in this work we compute a \textit{robust} control invariant set and a control Lyapunov function for the uncertain system. We underline that extending the nominal and certainty equivalent LMPC approaches~\cite{LMPC,LMPClinear, rosolia2017robust} to uncertain systems is not straightforward. 
In fact, as it will be clear later in this paper, 
standard shifting MPC arguments for proving recursive robust constraint satisfaction do not  apply. 
In particular, the use of different feedback policies along the prediction horizon is one of the key elements of the proposed control design which is necessary for providing robustness guarantees.


This paper is organized as follows: in Section~\ref{sec:background} we recall some definitions from set theory and the definition of Input-to-State Stability (ISS). Section~\ref{sec:invSetFromData} describes the challenges of learning safe sets from stored data of uncertain systems. Section~\ref{sec:ProbForm} describes the problem formulation and design requirements. The strategy proposed to compute safe sets and the $Q$-function is described in Section~\ref{sec:LMPCpre}. Section~\ref{sec:LMPCpolicy} describes the control design strategy and Section~\ref{sec:properties} illustrates the controller properties. A strategy to approximate the region of attraction of the controller is presented in Section~\ref{sec:regAtt}. Finally, in Section~\ref{sec:results} we show the effectiveness of the proposed controller an a double integrator system subject to bounded additive uncertainty.

\section{Technical Background}\label{sec:background}
In this section we recall some definitions from set theory \cite[Chapter 10]{BorrelliBemporadMorari_book}, which will be used later in this work.

\begin{definition}[Positive Invariant Set]\label{def:Inv} A set $\mathcal{O} \subseteq \mathcal{X}$ is said to be a positive invariant set for the autonomous system $x_{t+1} = A x_t$, if
\begin{equation*}
x \in \mathcal{O} \rightarrow Ax \in \mathcal{O}.
\end{equation*}
\end{definition}

\begin{definition}[Robust Positive Invariant Set]\label{def:RobustInv} A set $\mathcal{O} \subseteq \mathcal{X}$ is said to be a robust positive invariant set for the uncertain autonomous system $x_{t+1} = A x_t + w_t$, with $w_t \in \mathcal{W}$ 
if
\begin{equation*}
x \in \mathcal{O} \rightarrow Ax+w\in \mathcal{O},~ \forall w \in \mathcal{W}.
\end{equation*}
\end{definition}

\begin{definition}[Robust Control Positive Invariant Set]\label{def:RobustInv} A set $\mathcal{C} \subseteq \mathcal{X}$ is said to be a robust control positive invariant set for the uncertain system $x_{t+1} = A x_t + Bu_t+ w_t$, with $w_t \in \mathcal{W}$ and $u_t \in \mathcal{U}$, if
\begin{equation*}
x \in \mathcal{C} \rightarrow \exists u \in \mathcal{U}: A x + B u + w \in \mathcal{C},~ \forall w \in \mathcal{W}.
\end{equation*}
\end{definition}

\begin{definition}[Robust Successor Set]\label{def:RobustSuc}
Given a control policy $\pi(\cdot)$ and the closed-loop system $x_{t+1} = A x_t + B\pi(x_t) + w_t$, we denote the robust successor set from the set $\mathcal{S}$ as 
\begin{equation*}
\begin{aligned}
    \Succ(\mathcal{S}, \mathcal{W}, {\pi}) = \{ & x_{t+1} \in \mathbb{R}^n :  \exists x_t \in \mathcal{S},  \exists w_t \in \mathcal{W}  \\
    & \text{ such that } x_{t+1} = A x_t + B \pi(x_t) + w_t \}.
\end{aligned}
\end{equation*}
\end{definition}

Given the initial state $x_t$, the robust successor set $\Succ( x_t , \mathcal{W}, {\pi})$ collects the states that the uncertain autonomous system may reach in one time step.

\begin{definition}[N-Step Robust Reachable Set]\label{def:RobustSuc} Given a \\ control policy $\pi(\cdot)$ and the closed-loop system $x_{t+1} = A x_t + B\pi(x_t) + w_t$ with $w_t \in \mathcal{W}$ for all $t\geq 0$, for $k=\{0, \ldots, N-1 \}$ we recursively define the $N$-step robust reachable set from the set $\mathcal{S}$ as 
\begin{equation*}
\begin{aligned}
    \mathcal{R}_{t\rightarrow t+k+1}(\mathcal{S}, {\pi}) = \Succ(\mathcal{R}_{t\rightarrow t+k}(\mathcal{S}, {\pi}), \mathcal{W}, {\pi})
\end{aligned}
\end{equation*}
where $\mathcal{R}_{t\rightarrow t}(\mathcal{S}, {\pi})=\mathcal{S}$. Robust reachable sets are also referred to as forwards reachable sets.
\end{definition}

Given a linear time-invariant system, the $N$-Step robust reachable set $\mathcal{R}_{t\rightarrow t+N}(\mathcal{S}, {\pi})$ collects the states which can be reached from the set $\mathcal{S}$ in $N$-steps.
Finally, we recall the definition of Input-to-State Stability (ISS) of a robust invariant set \cite{lin1995various}, which extends the more widely known notion of ISS of an equilibrium point \cite{jiang2001input, goulart2006optimization,khalil2002nonlinear,grune2014iss}. We use the standard function classes $\mathcal{K}$, $\mathcal{K}_\infty$, and $\mathcal{KL}$ notation (see \cite{kellett2014compendium}) and we define the distance from a point $x\in\mathbb{R}^n$ to a set $\mathcal{O}\subseteq \mathbb{R}^n$ as
\begin{equation*}
    |x|_\mathcal{O} \overset{\Delta}{=} \inf_{d\in\mathcal{O}} ||x-d||_2.
\end{equation*}

\begin{definition}[Input to State Stability (ISS)~\cite{lin1995various}]\label{def:iss}
Let $\mathcal{O}$ be a robust positive invariant set for the autonomous system~$x_{t+1} = A x_t + B \pi (x_t) + w_t$ with $w_t \in \mathcal{W}$. We say that the closed-loop system is ISS with respect to $\mathcal{O}$ if for all $w_t \in \mathcal{W}$, $t \geq 0$, and $x_0 \in \mathcal{X}$ we have:
\begin{equation*}
    |x_{t+1}|_\mathcal{O} \leq \beta(|x_0|_{\mathcal{O}}, t+1) + \gamma \big(  \textstyle \text{sup}_{k \in \{0, \ldots, t\}}  ||w_k|| \big),
\end{equation*}
where $\beta(\cdot, \cdot)$ is a class-$\mathcal{KL}$ function and $\gamma(\cdot)$ is a class-$\mathcal{K}$ function.
\end{definition}


\section{Computing Invariant Sets From Data}\label{sec:invSetFromData}
In order to motivate our work, this section highlights the challenges associated with the computation of invariant sets from data. First, we recall from~\cite{LMPC, LMPClinear} how historical data can be used to compute invariant sets for deterministic systems. Consider the discrete time linear system
\begin{equation*}
    \bar x_{t+1}^j = A \bar x_t^j + B \pi^j(\bar x_t^j)
\end{equation*}
where $\pi^j(\cdot)$ is a feedback policy known only along the $j$th stored trajectory $ \bar {\bf{x}}^j = [\bar x_0^j, \ldots, \bar x_t^j, \ldots, \bar x_{T^j}^j]$. 
Assume that $\pi^j(\cdot)$ is able to execute the desired task safely, meaning that $\bar x_{T^j} \in \mathcal{O}$. 
At any iteration $i>j$ and time $k\geq 0$, if the system state $x_k^i$ equals a state $x_t^j$ which has been visited at the previous  $j$th iteration, then the feedback policy $\pi^j(\cdot)$ will drive the system along the $j$th trajectory. This obvious fact is a consequence of the system being deterministic. More importantly, as the policy $\pi^j(\cdot)$ brings the system to the invariant set $\mathcal{O}$, the convex hull of visited states and $\mathcal{O}$ is a control invariant set.
Therefore, invariant sets  for deterministic systems can be easily built from data.


In contrast, when dealing with \emph{uncertain systems}, the set of visited states is not an invariant set. In fact, consider the discrete time uncertain system
\begin{equation*}
    x^j_{t+1} = A x^j_t + B \pi^j(x^j_t) + w_t^j,
\end{equation*}
where the random disturbance $w_t^j$ belongs to the set  $\mathcal{W}$ and the $j$th stored trajectory is ${\bf{x}}^j = [x_0^j, \ldots, x_t^j, \ldots, x_{T^j}^j]$. 
Assume that $\pi^j(\cdot)$ is able to execute the desired task safely at iteration~$j$. 
Notice that the stored trajectory ${\bf{x}}^j$ is associated with a specific disturbance realization $[w_0^j,\ldots,w_t^j,\ldots]$. For this reason at any iteration $i>j$ and time $k\geq 0$, if the system state $x_k^{i}$ equals a state $x_t^j$ that has been visited, applying the feedback policy $\pi^j(\cdot)$ may drive the system to a state neither stored nor safe, due to a potentially different disturbance realization $[w_0^{i},\ldots,w_t^{i},\ldots]$. In conclusion, the set of visited states cannot be naively exploited to compute invariant sets.

Furthermore, we underline that even if a control invariant set for uncertain systems can be computed from data, its use for MPC design is not straightforward. This issue will become clear later in this paper. In the following, we first present a strategy to construct robust control invariant sets using stored data from a linear uncertain system. Afterwards, we leverage these sets to iteratively synthesize robust LMPC policies.


\section{Problem Formulation}\label{sec:ProbForm}
We consider a discrete time uncertain system of the form:
\begin{equation}\label{eq:system0}
    x_{t+1} = A x_t + B u_t + w_t,
\end{equation}
where $x_t \in \mathbb{R}^n$ and  $u_t \in \mathbb{R}^d$ are the state and the input at time $t$ and the matrices $A$ and $B$ are known. The disturbances $w_t^j$ are zero mean independent and identically distributed (\textit{i.i.d.}) with bounded support $\mathcal{W}\subset\mathbb{R}^n$.

\begin{assumption}\label{ass:first}
The disturbance's support $\mathcal{W}$ is a compact polytope described by $l$ vertices $\{v^1_w, \ldots, v^l_w\}$ and it contains the origin.
\end{assumption}

Furthermore, the system is subject to the following convex constraints on states and inputs:
\begin{equation}\label{eq:constraintSet}
		x_t \in \mathcal{X} \text{ and } u_t \in \mathcal{U}, ~ \forall t \geq 0,
\end{equation}
where the sets $\mathcal{X}$ and $\mathcal{U}$ contain the origin and are assumed to be compact.

In this paper, we consider control tasks where we would like to steer system~\eqref{eq:system0} towards the set $\mathcal{O}$, while minimizing the summation of the stage cost $h:\mathbb{R}^n \times \mathbb{R}^d \rightarrow \mathbb{R}$ and satisfying constraints~\eqref{eq:constraintSet}. Throughout the paper we make the following assumptions.
\begin{assumption}\label{ass:O_inf}
The set $\mathcal{O} \subset \mathcal{X} \subset \mathbb{R}^n$ is a robust positive invariant set for the autonomous system $x_{t+1} = (A + BK) x_t+w_t$ with $w_t \in \mathcal{W}$. Furthermore, $\mathcal{O}$ is a polyhedron defined through its vertices $\{v_o^1, \ldots, v^{m}_o \}$ and
\begin{equation*}
    K \mathcal{O} = \{ u \in \mathbb{R}^d: \exists x \in \mathcal{O}, u = Kx \}.
\end{equation*}
\end{assumption}
\begin{assumption}\label{ass:cost}
The  stage cost $h(\cdot,\cdot)$ is continuous and jointly convex in its arguments. Furthermore, we assume that for all $x \in \mathbb{R}^n$, and for all $ u \in \mathbb{R}^d$ the stage cost satisfies:
\begin{equation*}
\begin{aligned}
\alpha_x^l(|x|_\mathcal{O}) \leq h(x,0) &\leq \alpha_x^u(|x|_\mathcal{O})\\
&\text{and } \alpha_u^l ( |u|_{K\mathcal{O}} ) \leq h(0,u) \leq \alpha_x^u(|u|_{K\mathcal{O}}),
\end{aligned}
\end{equation*}
where $\alpha_{x}^u,\alpha_{x}^l, \alpha_{u}^u$, and $\alpha_{u}^l \in \mathcal{K}_\infty$.
\end{assumption}

\vspace{0.2cm} \begin{remark}
In Assumption~\ref{ass:O_inf} a robust invariant $\mathcal{O}$ is  required.  In the proposed approach $\mathcal{O}$ can be a very small neighborhood of the origin. In fact, the iterative nature of the control design will allow us enlarge the closed-loop region of attraction at each iteration.

\end{remark}

\subsection{Iteration and Control Design Objectives}\label{sec:Obj}
Iteration $j$ refers to a finite time execution of the control task from the initial conditions $x_0^j$. At each iteration, the initial condition may be different, however the system model~\eqref{eq:system0}, the constraints~\eqref{eq:constraintSet}, the set $\mathcal{O}$, and the cost function $h(\cdot,\cdot)$ are identical. We denote $x_t^j$ and $u_t^j$ as the state and input of the system at time $t$ of iteration $j$, i.e.,
\begin{equation}\label{eq:sys}
	x_{t+1}^j = Ax_t^j + B u^j_t + w_t^j.
\end{equation}
Furthermore, we define $T^j$ as the finite time duration of iteration $j$. As we will discuss in Section~\ref{sec:LMPCpre}, the time $T^j$ is selected by the control designer.

At each iteration $j$, our objective is to design a state feedback policy for the uncertain system~\eqref{eq:sys}  
\begin{equation}\label{eq:feedbackPolicy}
    \pi^j(\cdot): \mathcal{C}^j \subseteq \mathbb{R}^n  \rightarrow \mathbb{R}^d,
\end{equation}
and the associated region of attraction $\mathcal{C}^j$ such that at each $j$th iteration and for all $x_0^j \in \mathcal{C}^j \subseteq \mathbb{R}^n$ we have that:
\begin{enumerate}
    \item The certainty equivalent system 
    \begin{equation}\label{eq:disturbanceFreeSys}
        \bar x_{t+1}^j = A \bar  x^j_t + B \bar  u^j_t
    \end{equation}
    with $\bar u^j_t = \pi^j(\bar x^j_t)$ converges asymptotically to the set $\mathcal{O}$, i.e. $\lim_{t \rightarrow \infty} \bar x_t^j \in \mathcal{O}$. 

    \item The closed-loop system $x_{t+1}^j = A x_t^j + B \pi^j(x_t^j) + w_t^j$ is Input-to-State Stable (ISS) with respect to the set $\mathcal{O}$ (see Section~\ref{sec:background} for the definition of ISS). 

    \item The closed-loop state and input constraints are robustly satisfied, namely
    \begin{equation*}
        x_t^j \in \mathcal{X} \text{ and } \pi^j(x_t^j) \in \mathcal{U }, \forall w_t^j \in \mathcal{W},~\forall t \geq 0.
    \end{equation*}
    
    \item The domain $\mathcal{C}^j$ of policy $\pi^j(\cdot)$ does not shrink with the iteration index, i.e., $\mathcal{C}^j \subseteq \mathcal{C}^{j+1}$.

    \item The iteration cost of the certainty equivalent system~\eqref{eq:disturbanceFreeSys}, defined as 
    \begin{equation*}
        J^{j}_{0\rightarrow T^j}(\bar x_0^{j}) = \sum_{k=0}^{T^j} h(\bar x_t^j,\pi^j(\bar x_t^j)),
    \end{equation*}
    is an upper-bound to the function $Q^{j-1}(\cdot)$ (i.e. $J^{j}_{0\rightarrow  T^j}( x_0^{j}) \leq Q^{j-1}( x_0^{j})$), which is non-increasing at each iteration 
    \begin{equation*}
        Q^{j}(\bar x) \leq Q^{k}(\bar x), \forall j \geq k.
    \end{equation*}

\end{enumerate}
Property 5) implies that, as more data is collected, the upper-bound on the performance of the certainty equivalent closed-loop system is non-increasing.

\vspace{0.2cm} \begin{remark}
Note that convergence and stability properties 1) and 2) hold for $t\rightarrow \infty$, although we study iterations with finite time duration. In particular, an iteration is terminated after $T^j$ time steps and, therefore, at the end of the $j$th iteration the closed-loop system may not reach the set $\mathcal{O}$. As we will discuss later on, the duration of an iteration $T^j$ may affect the closed-loop performance of the controller.
\end{remark}

\section{LMPC Preliminaries}\label{sec:LMPCpre}
In this section, we assume that a robust policy at iteration $j$ is given and we describe how to compute the terminal cost and constraints used to synthesize the control policy at the next iteration $j+1$.
In particular, we show how historical data and the $j$th robust MPC policy can be used to build a \textit{robust safe set} of states from which the iteration $j+1$ can be executed.
	
The iterative synthesis procedure used to update the control policy is described in Section~\ref{sec:LMPCpolicy}. This iterative procedure is based on the robust safe set, $Q$-function, and safe policy, which are defined  in this section. Their initialization at iteration $j=0$ is also discussed.

\subsection{Robust Safe Set}\label{sec:robSS}
This section shows how to iteratively construct robust control invariant sets. In particular, we run the closed-loop system at iteration $j$ and we exploit the closed-loop trajectory to construct a robust safe set.

We initialize the \textit{robust convex safe set} at iteration $j=0$ as
\begin{equation}\label{eq:SS_init}
	\mathcal{CS}^0=\mathcal{O}.
\end{equation}
Afterwards, we design the robust $N$-steps policy
\begin{equation}\label{eq:NstepPolicy}
    \boldsymbol{\pi}^{1,*}_t(\cdot)= [\pi^{1,*}_{t|t}(\cdot), \ldots, \pi^{1,*}_{t+N\text{-}1|t}(\cdot)],
\end{equation}
which steers the following closed-loop system
\begin{equation}\label{eq:NstepPredicted}
    x^1_{t+1}= A x^1_t +B \pi^{1,*}_{t|t}(x_t) + w_t
\end{equation}
to the robust convex safe set $\mathcal{CS}^0$. Next, we show that the control policy~\eqref{eq:NstepPolicy} and the convex safe set~\eqref{eq:SS_init} can be used to construct the convex safe set at iteration $1$. 
\begin{assumption}\label{ass:policyAssumption}
For all $t \in \{0, \ldots, T^1\}$, the $N$-steps policy $\boldsymbol{\pi}^1_t(\cdot)$ in~\eqref{eq:NstepPolicy} robustly steers the predicted closed-loop system
\begin{equation*}
    \begin{aligned}
        x_{k+1|t}^1 = A x_{k|t}^1 + B  \pi_{k|t}^{1,*}( x_{k|t}^1 & ) + w^1_{k|t}, \\
    &\forall k = t, \ldots, t+N-1
    \end{aligned}
\end{equation*}
from the state $x_t^1{\in \mathcal{C}^1}$ to the robust convex safe set $\mathcal{CS}^0=\mathcal{O}$ in $N$-steps, while robustly satisfying state and input constraints~\eqref{eq:constraintSet}.
\end{assumption}

	\vspace{0.2cm} \begin{remark}
			The control policy $\boldsymbol{\pi}^{1,*}_t(\cdot)$, which satisfies Assumption~\ref{ass:policyAssumption}, can be computed using the iterative procedure that we will be describing in Section~\ref{sec:LMPCpolicy}. 
	\end{remark}

\vspace{0.2cm} \begin{remark}
We underline that $T^j$ represents the duration of the $j$th closed-loop simulation. On the other hand, $N$ is the length of the prediction horizon associated with the control policy $\boldsymbol{\pi}^{1,*}_t(\cdot)$. Notice that $T^j$ and $N$ are not related but in general $N<T^j$. Indeed, the horizon length $N$ is usually chosen much smaller than the task duration $T^j$ to reduce the computational burden associated with the control policy. Finally, we underline that the duration of the control task $T^j$ is chosen by the control designer. 
\end{remark}

Let the vectors
\begin{equation}\label{eq:recordedUncertainTrajectory}
\begin{aligned}
 	 [x_0^{1}, \ldots,  x_{T^1}^{1}] \\
	 [u_0^{1}, \ldots,  u_{T^1}^{1}]
\end{aligned}
\end{equation}
collect states and inputs associated with a simulation of the closed-loop system~\eqref{eq:NstepPredicted}. As we will discuss later on, by linearity of the system, any state in the convex-hull of the closed-loop trajectory in~\eqref{eq:recordedUncertainTrajectory} can be robustly steered to~$\mathcal{O}$.
However, the convex hull of the stored states in~\eqref{eq:recordedUncertainTrajectory} and $\mathcal{O}$, denoted as $\text{Conv}\{\{\cup_{t=0}^{T^j}x_t^1\}\cup\mathcal{O} \}$, is not invariant. In fact, the set $\text{Conv}\{\{\cup_{t=0}^{T^j}x_t^1\}\cup\mathcal{O} \}$ does not necessarily contain the $k$-steps robust reachable sets $\mathcal{R}_{t \rightarrow t + k}( x_0^1, { \boldsymbol{\pi}^{1,*}_t})$ from the starting state $x_0^1$, as shown in Figure~\ref{fig:stateNominalEvolution}. 
\begin{figure}[h!]
    \centering
    \includegraphics[width= 0.9\columnwidth]{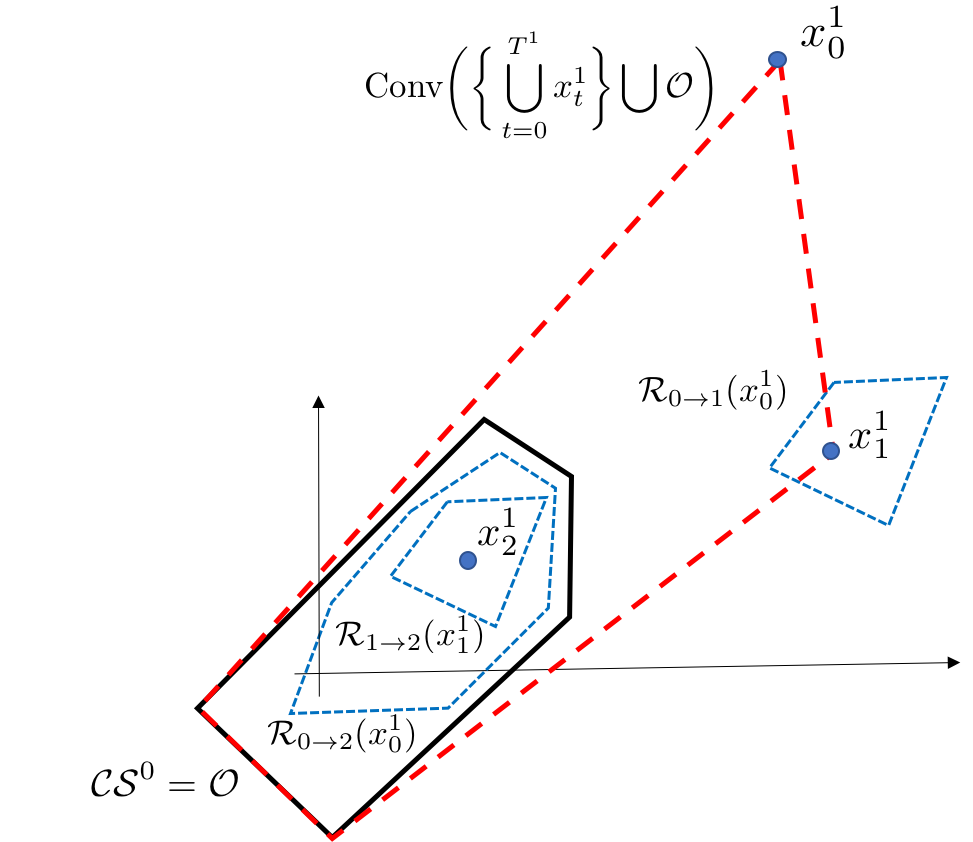}
    \caption{The figure shows the robust invariant set $\mathcal{O}$ (solid black), the robust reachable sets $\mathcal{R}_{t\rightarrow t+k}(x_0^j,  \boldsymbol{\pi}^{1,*}_t)$ (dashed blue line) and the convex hull of the states in~\eqref{eq:recordedUncertainTrajectory} and $\mathcal{O}$, denoted as $\text{Conv}\{\{\cup_{t=0}^{T^j}x_t^1\}\cup\mathcal{O} \}$ (dashed red line). We notice that the set $\text{Conv}\{\{\cup_{t=0}^{T^j}x_t^1\}\cup\mathcal{O} \}$ (dashed red line) does not contain the robust reachable sets $\mathcal{R}_{t\rightarrow t+k}(x_0^j,  \boldsymbol{\pi}^{1,*}_t)$ (dashed blue line) and therefore it is not a robust invariant for the closed-loop system~\eqref{eq:NstepPredicted}.}
    \label{fig:stateNominalEvolution}
\end{figure}

Now, we notice that robust control invariant sets can be computed using $k$-steps robust reachable sets $\mathcal{R}_{t\rightarrow t+k}( x_0^1, { \boldsymbol{\pi}^{1,*}_t})$ from the stored states in~\eqref{eq:recordedUncertainTrajectory}. In particular, we notice that the union of the $k$-steps robust reachable sets $\mathcal{R}_{t\rightarrow t+k}( x_0^1, { \boldsymbol{\pi}^{1,*}_t})$ for $k=0, \ldots, N$ and the robust convex safe set $\mathcal{CS}^0=\mathcal{O}$ is a robust control invariant. Therefore, we define the \textit{robust convex safe set} at iteration $j=1$ as the convex hull of the robust reachable sets $ \mathcal{R}_{t\rightarrow t+k}(x_t^1, { \boldsymbol{\pi}^{1,*}_t})$ and the robust convex safe set $\mathcal{CS}^0$ at iteration~$0$:
\begin{equation}
    \mathcal{CS}^1 = \text{Conv}\Bigg(  \bigg\{ \bigcup_{t=0}^{T^1} \bigcup_{k=0}^N \mathcal{R}_{t\rightarrow t+k}(x_t^1, { \boldsymbol{\pi}^{1,*}_t}) \bigg\}\bigcup \mathcal{CS}^0\Bigg).
\end{equation}
The above robust convex safe set at iteration $j=1$ is constructed for the closed-loop system~\eqref{eq:NstepPredicted} and it is shown in Figure~\ref{fig:cs}.

\begin{figure}[h!]
    \centering
    \includegraphics[width= 0.9\columnwidth]{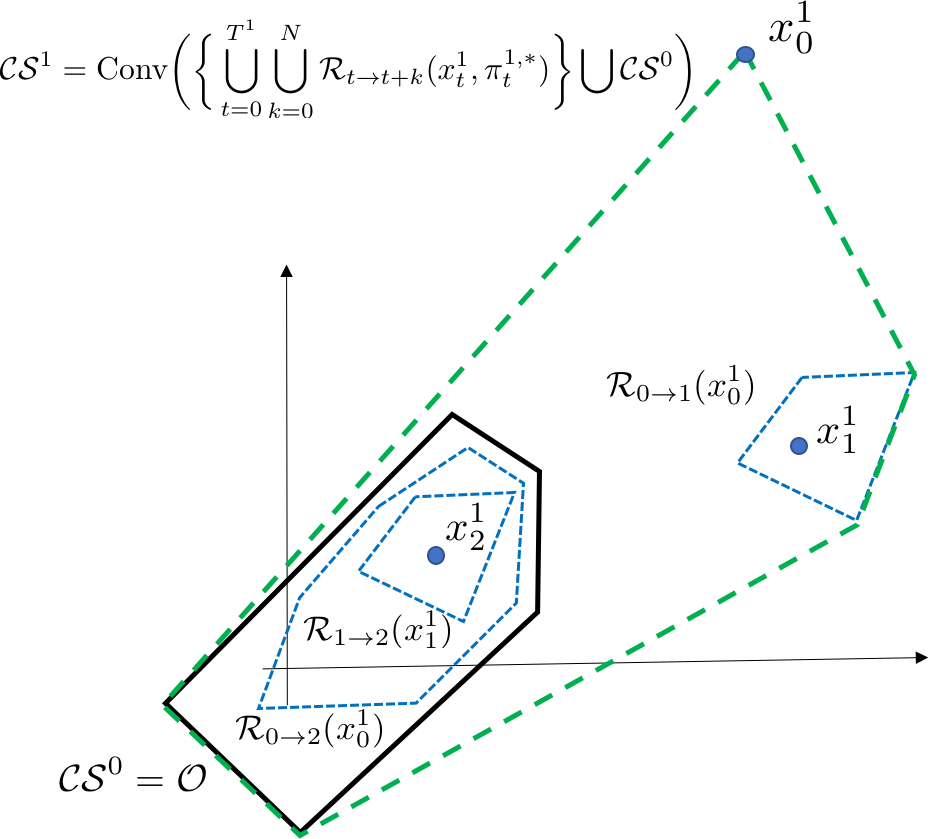}
    \caption{The figure shows the robust invariant set $\mathcal{O}$ (solid black line), the convex safe set $\mathcal{CS}^1 = \text{Conv}( \{ \bigcup_{t=0}^{T^1} \bigcup_{k=0}^N \mathcal{R}_{t\rightarrow t+k}(x_t^1, { \boldsymbol{\pi}^{1,*}_t}) \} \bigcup \mathcal{CS}^0)$ (dashed green line) and the robust reachable sets $\mathcal{R}_{t\rightarrow t+k}(x_0^1, { \boldsymbol{\pi}^{1,*}_t})$ (dashed blue line).}
    \label{fig:cs}
\end{figure}

The above process is repeated at each iteration $j\geq 1$. Clearly, Assumption~\ref{ass:policyAssumption} must hold when $\mathcal{CS}^{0}$ is replaced with $\mathcal{CS}^{j-1}$ and iteration 1 with $j$. 
More formally, given the $N$-steps policy 
\begin{equation}\label{eq:policy_j_iteration}
    \boldsymbol{\pi}^{j,*}_t(\cdot) = [\pi_{t|t}^{j,*}(\cdot), \ldots, \pi_{t+N\text{-}1|t}^{j,*}(\cdot)]
\end{equation}
and the closed-loop system
\begin{equation}\label{eq:cl_sys}
    x_{t+1}^j = Ax_{t}^j + B \pi_{t|t}^{j,*}(x_{t}^j) + w_t^j
\end{equation}
we assume that the following holds.
\begin{assumption}\label{ass:policyAss}
For all $t \in \{0, \ldots, T^j\}$, the $N$-steps policy $\boldsymbol{\pi}^{j,*}_t$ from~\eqref{eq:policy_j_iteration} robustly steers the predicted closed-loop system
\begin{equation*}
    \begin{aligned}
        x_{k+1|t}^j = A x_{k|t}^j + B  \pi_{k|t}^{j,*}( x_{k|t}^j & ) + w^j_{k|t}, \\
    &\forall k = t, \ldots, t+N-1
    \end{aligned}
\end{equation*}
from the state $x_t^j {\in \mathcal{C}^j}$ to the robust convex safe set $\mathcal{CS}^{j-1}$ in $N$-steps, while robustly satisfying state and input constraints~\eqref{eq:constraintSet}.
\end{assumption}

Later in Section~\ref{sec:LMPCpolicy}, we will show how to synthesize a control polity $\boldsymbol{\pi}^{j,*}_t$ which satisfies Assumption~\ref{ass:policyAss}. 

At iteration $j$, we iteratively define the \textit{robust convex safe set}: 
\begin{equation}\label{eq:CS}
    \mathcal{CS}^j = \text{Conv}\Bigg(  \bigg\{ \bigcup_{t=0}^{T^j} \bigcup_{k=0}^N \mathcal{R}_{t\rightarrow t+k}(x_t^j, { \boldsymbol{\pi}^{j,*}_t}) \bigg\}\bigcup \mathcal{CS}^{j-1}\Bigg).
\end{equation}
Details on the computation and storage of the convex safe set are provided next.

\subsection{Robust Convex Safe Set: Vertex Representation}
Recall from Assumption~\ref{ass:O_inf} that $l$ denoted the number of vertices of the disturbance support. Now, we define the $l^k$ vertices of the $k$-step robust reachable set $\mathcal{R}_{t\rightarrow t+k} (x^j_t, { \boldsymbol{\pi}^{j,*}_t})$ from $x^j_t$ as
\begin{equation}\label{eq:verticesReachSet}
    [v_{t+k|t}^{j,1}, \ldots, v_{t+k|t}^{j,l^k}].
\end{equation}
The vertices of the robust reachable sets $\mathcal{R}_{t\rightarrow t+k} (x^i_t, { \boldsymbol{\pi}^{i,*}_t})$ for all $k \in \{0,\ldots, N-1\}$, $i\in\{0,\ldots,j\}$ and $t \in \{0, \ldots, T^j\}$ are collected by the following matrix
\begin{equation}\label{eq:xMatrix}
\begin{aligned}
 	{\bf{X}}^j = [{\bf{X}}^{j-1},  &~v_{0|0}^{j,1},\ldots,  v_{N-1|0}^{j,l^{N-1}}, \ldots\\
 	&~v_{t|t}^{j,1}, \ldots,  v_{t+N-1|t}^{j,l^{N-1}}, \ldots\\
 	&~v_{T^j|T^j}^{j,1}, \ldots,  v_{T^j + N-1|T^j}^{j,l^{N-1}}],
\end{aligned}
\end{equation}
where at the $j$th iteration $v_{t+k|t}^{j,i}$ represents the $i$th vertex of the robust reachable set $\mathcal{R}_{t\rightarrow t+k} (x^j_t, { \boldsymbol{\pi}^{j,*}_t})$.
In the above recursive definition, we set ${\bf{X}}^{0}=[v_o^1, \ldots, v_{o}^m]$, where $v_o^i$  for $i\in\{1,\ldots,m\}$ are the vertices of $\mathcal{O}$ from Assumption~\ref{ass:O_inf}. The matrix ${\bf{X}}^j \in \mathbb{R}^{n \times \text{col}({\bf{X}}^j)}$ , where the number of columns $\text{col}({\bf{X}}^j) = (T^j +1)\sum_{k=0}^{N-1}l^k +  \text{col}({\bf{X}}^{j-1})$.

Finally, as the columns of the matrix ${\bf{X}}^j$ in~\eqref{eq:xMatrix} collect all vertices of the robust reachable sets $\mathcal{R}_{t\rightarrow t+k}(x_t^j, { \boldsymbol{\pi}^{j,*}_t})$, the robust convex safe set $\mathcal{CS}^j$ from \eqref{eq:CS} can be written as
\begin{equation}\label{eq:RS_vertex}
    \mathcal{CS}^j = \Big\{ x \in \mathbb{R}^n : \exists {\bm{\lambda}}^j \geq 0, {\bf{X}}^j {\bm{\lambda}}^j = x \text{ and } \mathds{1}^\top{\bm{\lambda}}^j = 1 \Big\},
\end{equation}
where $\mathds{1} \in \mathbb{R}^{\text{col}({\bf{X}}^j)}$ is a vector of ones.

\vspace{0.2cm} \begin{remark}
Notice that the strategy presented in this paper is based on a commonly used ``vertex enumeration approach"~\cite{scokaert1998min, BorrelliBemporadMorari_book}. Its worst case complexity is exponential in the horizon $N$ of the feedback policy~\eqref{eq:policy_j_iteration}, although independent on the length of the task duration $T^j$. We underline that this paper focuses on the fundamental properties of the controller design. Computational tractability can be obtained as in any MPC scheme by using a different disturbance model or feedback parametrization.
\end{remark}


\subsection{Robust Q-Function}\label{sec:itCost}
The robust $Q$-function approximates the cost-to-go over the robust convex safe set and it is constructed iteratively as explained next. 
At iteration $j$, we assume that we are given the robust $Q$-function  $Q^{j-1}(\cdot)$ which maps each state $x \in \mathcal{CS}^{j-1}$ to the closed-loop cost, and we show how to construct a robust $Q$-function at the next iteration $j$.
This recursion is initialized at iteration $0$ setting the robust $Q$-function 
\begin{equation}\label{eq:Qfun_init}
Q^0(x) = 0, ~\forall x \in \mathcal{CS}^0 = \mathcal{O}.
\end{equation}

While in the nominal case from~\cite{LMPClinear} the vertices of the convex safe set are a subset of the stored trajectory, the convex safe set from~\eqref{eq:CS} may introduce additional vertices representing the worst case predicted realizations. For this reason, a cost-to-go associated with such predicted worst case realizations should be defined. In the following we define the cost-to-go $J^{j}_{t|t}$ associated with the stored states $x_{t|t}^j = x_t^j$ and the predicted cost-to-go $J^{j}_{k|t}$ associated with the predicted state $x^{j}_{k|t}$ at time $k$. In particular, after completion of the iteration $j$ for $t\in \{0, \ldots, T^j\}$, $k \in \{0, \ldots, N-1\}$ and $i \in \{1,\ldots,l^{k}\}$, we compute the cost-to-go for the vertices $v_{k|t}^{j,i}$ of $\mathcal{CS}^j$ from ${\bf{X}}^j$ in~\eqref{eq:xMatrix} as
\begin{equation}\label{eq:J_definition1}
\begin{aligned}
J^{j}_{k|t} \big( v_{k|t}^{j,i}\big)\! = \! \min_{\boldsymbol{\gamma}_k \geq 0}~~ & h\big( v_{k|t}^{j,i}, \pi_{k|t}^{j,*} (v_{k|t}^{j,i} ) \big) \! + \! \sum_{r=1}^{l^{k+1}} \gamma^r J^{j}_{k+1|t} ( v_{k+1|t}^{j,r} ),\\
\text{s.t.}~~~~ & \sum_{r=1}^{l^{k+1}} \gamma^r  v_{k+1|t}^{j,r} = A v_{k|t}^{j,i} + B  \pi_{k|t}^{j,*} (v_{k|t}^{j,i})\\
& \sum_{r=1}^{l^{k+1}} \gamma^r = 1
\end{aligned}
\end{equation}
where $\boldsymbol{\gamma}_k = [\gamma^1,\ldots, \gamma^{l^{k+1}}]$ and ${\pi}^{j,*}_{k|t}(\cdot)$ is the control policy from~\eqref{eq:policy_j_iteration}. In the above recursion we set
\begin{equation}\label{eq:J_definition2}
    J^{j}_{t+N |t }\big( v_{t+N|t}^{j,i} \big) = Q^{j-1}\big( v_{t+N|t}^{j,i} \big),~\forall i \in \{1,\ldots, l^{N}\}.
\end{equation} 
Basically, the cost-to-go $J^{j}_{k|t} ( v_{k|t}^{j,i})$ at time $k$ is computed summing up the running cost and the interpolated cost-to-go at the next predicted time $k+1$.

Given $Q^{j-1}(\cdot)$, the cost-to-go $J^{i}_{k|t} ( \cdot )$ is computed for all $i\in \{0, \ldots, j\}$, $t \in \{0,\ldots, T^j\}$ and $k\in \{0, \ldots, N\text{-}1\}$. Then, these cost values are collected in the following vector
\begin{equation*}
\begin{aligned}
    {\bf{J}}^{j} = [ {\bf{J}}^{j-1}, &~J_{0|0}^j\big(v_{0|0}^{j,1} \big),\ldots,  J_{N-1|0}^j\big(v_{N-1|0}^{j,l^{N-1}} \big), \ldots\\
 	&~J_{t|t}^j\big(v_{t|t}^{j,1}\big), \ldots,  J_{t+N-1|t}^j\big(v_{t+N-1|t}^{j,l^{N-1}} \big), \ldots\\
 	&~J_{T^j|T^j}^j\big(v_{T^j|T^j}^{j,1} \big), \ldots,  J_{T^j+N-1|T^j}^j\big(v_{T^j+N-1|T^j}^{j, l^{N-1}}\big)],
\end{aligned}
\end{equation*}
where ${\bf{J}}^{0} = [0, \ldots, 0]$ represents the cost-to-go associated with the vertices of $\mathcal{O}$.
Finally, we define the $Q$-function at iteration $j$ which interpolates the cost-to-go over the robust safe set,
\begin{equation}\label{eq:Qfun}
	\begin{aligned}
		Q^j(x) = \min_{ {\bm \lambda}^j \in \Lambda^{j}(x)} & \quad {\bf{J}}^{j} {\bm \lambda}^j,  \\
	\end{aligned}
\end{equation}
where for the matrix ${\bf{X}}^j$ the set
\begin{equation}\label{eq:LambdaSet} 
\begin{aligned}
     \Lambda^{j}(x) = \Big\{{\bm{\lambda}}^j \in  \mathbb{R}^{\text{col}({\bf{X}}^j)} : {\bm{\lambda}}^j \geq 0, ~&{\bf{X}}^j{\bm{\lambda}}^j = x \\&\text{ and } \mathds{1}^\top{\bm{\lambda}}^j = 1 \Big\}
\end{aligned}
\end{equation}
collects the vectors ${\bm{\lambda}}^j$ which can be used to express $x$ as a convex combination of the columns of ${\bf{X}}^j$.

\vspace{0.2cm} \begin{remark}
Notice that a longer task duration $T^j$ results in more data points used to compute the cost vector ${\bf{J}}^{j}$. Consequently, the $Q$-function from~\eqref{eq:Qfun} can only decrease when the duration of the control task $T^j$ is increased. As we will discuss later on, the $Q$-function from~\eqref{eq:Qfun} is an upper-bound of the closed-loop performance of the certainty equivalent system, therefore the upper-bound on the cost can only decrease as the duration of the control task $T^j$ increases.
\end{remark}

\subsection{Set of Safe Policies}\label{sec:dataBasedPolicy}

At this point we have shown how to compute a robust control invariant set and a $Q$-function based on data collected at the $j$th iteration. The last missing element needed for an MPC design is the feedback controller associated with the terminal set. 

Here we show how to construct a set of safe policies $\mathcal{SP}^j$, which may be used to robustly constrain the evolution of system~\eqref{eq:sys} within $\mathcal{CS}^j$, while satisfying state and input constraints~\eqref{eq:constraintSet}. 
We begin by presenting an implicit parametrization of the set of policies $\mathcal{SP}^j$ which is amenable for optimization and it can be used to design a predictive controller that guarantees recursive constraint satisfaction. Afterwards, we define a safe policy $\kappa^{j,*}(\cdot) \in \mathcal{SP}^j$, which is able to complete the task from any state within the robust convex safe set $\mathcal{CS}^j$.

First, we define the matrix ${\bf{U}}^j$ collecting the inputs associated with the data stored in~\eqref{eq:xMatrix},
\begin{equation}\label{eq:uMatrix}
\begin{aligned}
    {\bf{U}}^j \! = \! [{\bf{U}}^{j-1}, &   \pi^{j,*}_{0|0} \big( v_{0|0}^{j,1}\big),\ldots,  \pi^{j,*}_{N-1|0} \big(  v_{N-1|0}^{j,l^{N{-}1}}\big), \ldots\\
 	& \pi^{j,*}_{t|t} \big( v_{t|t}^{j,1}\big), \ldots,   \pi^{j,*}_{t+N-1|t} \big( v_{t+N-1|t}^{j,l^{N-1}}\big), \ldots\\
 	& \pi^{j,*}_{T^j|T^j} \big( v_{T^j|T^j}^{j,1} \big), \ldots,   \pi^{j,*}_{T^j+N-1|T^j} \big( v_{T^j+N-1|T^j}^{j,l^{N-1}}\big)]
\end{aligned}
\end{equation}
where the policies $ \pi^{j,*}_{k|t}$ are defined in~\eqref{eq:policy_j_iteration}.
In the above recursive definition ${\bf{U}}^{0}= [Kv_o^1, \ldots, Kv_{o}^m]$, where the feedback gain $K$ and the vertices $v_o^i$ of $\mathcal{O}$ are defined in Assumption~\ref{ass:O_inf}. 

Now, we notice that by linearity of system \eqref{eq:sys}, if a state $x \in \mathcal{CS}^j$ is expressed as a convex combination of the stored states $x = {\bf{X}}^j {\bm{\lambda}}^j$, then the input $u = {\bf{U}}^j {\bm{\lambda}}^j \in \mathcal{U}$ will keep the evolution of the system within $\mathcal{CS}^j$ for all disturbance realizations. More formally,  given the set $\Lambda^j(\cdot)$ defined in~\eqref{eq:LambdaSet}, we have that 
$\forall x \in \mathcal{CS}^j \subseteq \mathcal{X}, \forall {\bm{\lambda}}^j \in \Lambda^j(x)$
\begin{equation*}
   {\bf{U}}^j {\bm{\lambda}}^j \in \mathcal{U} \text{ and } Ax+B {\bf{U}}^j {\bm{\lambda}}^j +w \in \mathcal{CS}^j, \forall w \in \mathcal{W}.
\end{equation*} 
Therefore, the set of feedback policies $\kappa^j(\cdot): \mathbb{R}^n \rightarrow \mathbb{R}^d$
\begin{equation}\label{eq:feasPolicy}
\begin{aligned}
    \mathcal{SP}^j = \big\{  \kappa^j(\cdot) : \forall x \in \mathcal{CS}^j, \exists & {\bm{\lambda}}^j \in \Lambda^j(x), \\ &\text{ such that }\kappa^j(x) ={\bf{U}}^j {\bm{\lambda}}^j \big\},
\end{aligned}
\end{equation}
guarantees that $\forall \kappa^j(\cdot) \in \mathcal{SP}^j$ the robust convex safe set $\mathcal{CS}^j$ is a robust positive invariant set for the closed-loop system $x_{t+1} = Ax_t+B\kappa^j(x_t)+w_t$. This statement is formalized by the following Proposition~\ref{prop:invariance}.
\begin{proposition}\label{prop:invariance}
Let Assumptions~\ref{ass:first}--\ref{ass:policyAss} hold. Then, for all control policies $\kappa^j(\cdot) \in \mathcal{SP}^j$ and $\forall x \in \mathcal{CS}^j$ we have that
\begin{equation*}
     Ax+B\kappa^j(x)+w \in \mathcal{CS}^j\subseteq \mathcal{X}~ \forall w \in \mathcal{W}
\end{equation*}
and $\kappa^j(x) \in \mathcal{U}$.
\end{proposition}

\begin{proof}
The proof can be found in the Appendix.
\end{proof}

Finally, we define the safe policy
\begin{equation}\label{eq:theSafePolicy}
    \kappa^{j,*}(x) = {\bf{U}}^j {\bm{\lambda}}^{j,*}(x),
\end{equation}
where ${\bm{\lambda}}^{j,*}(x)$ is the minimizer in~\eqref{eq:Qfun}. Basically, the above safe policy evaluated at $x$ is given by the convex combination of stored inputs, for the multipliers ${\bm{\lambda}}^{j,*}(x)$ which define the robust $Q$-function at $x$. In the following propositions, we show that the $Q$-function is a Lyapunov function for the certainty equivalent closed-loop system~\eqref{eq:disturbanceFreeSys} and \eqref{eq:theSafePolicy}. Furthermore, we show that the policy~\eqref{eq:theSafePolicy} in closed-loop with system~\eqref{eq:sys} guarantees Input-to-State Stability (ISS).

\begin{proposition}\label{prop:QfunDecrease}
Let Assumptions~\ref{ass:first}--\ref{ass:policyAss} hold. Consider the $Q$-function $Q^j(\cdot)$ in~\eqref{eq:Qfun}, we have that for all $x \in \mathcal{CS}^j$
\begin{equation}
    Q^j(x) \geq h(x, \kappa^{j,*}(x)) + Q^j(Ax+B \kappa^{j,*}(x)),
\end{equation}
where $\kappa^{j,*}(\cdot)$ is the safe policy defined in~\eqref{eq:theSafePolicy}.
\end{proposition}

\begin{proof}
The proof can be found in the Appendix.
\end{proof}

\begin{proposition}
\label{prop:safePolicyISS}
Consider system~\eqref{eq:sys} in closed-loop with the safe policy~\eqref{eq:theSafePolicy}. Let Assumptions~\ref{ass:first}--\ref{ass:policyAss} hold and assume that $x_0 \in \mathcal{CS}^j$, then the closed-loop system~\eqref{eq:sys} and \eqref{eq:theSafePolicy} is input to state stable for the robust positive invariant set $\mathcal{O}$.
\end{proposition}

\begin{proof}
The proof can be found in the Appendix.
\end{proof}


\section{Learning Model Predictive Control}\label{sec:LMPCpolicy}
This section introduces the iterative control design procedure. 
At the end of iteration $j-1$, we collect a data set of costs, inputs, and states, which are used to construct the robust convex safe set and robust $Q$-function at iteration $j-1$, as described in Section~\ref{sec:LMPCpre}. Finally, we leverage these quantities to design a robust Learning Model Predictive Controller (LMPC) for the $j$th iteration. The LMPC policy is able to safely execute the control task at the $j$th iteration and it can be used to collect new closed-loop data to design the controller at the next iteration $j+1$.

\subsection{Policy Synthesis}
In this section, we introduce the LMPC policy. 
For more details on the control design choices and the differences with standard MPC design, we refer to the discussion in Section~\ref{sec:discussion} and to the properties description in Section~\ref{sec:properties}.

We define the following optimal control problem for the state $x_t^j \in \mathbb{R}^n$ and the parameter $N_t^j \in \{0,\ldots, N\}$:
\begin{subequations}\label{eq:FTOCP}
\begin{align}
&\!\!\!\!\!\!\!\!\!\!\!\!\!\!\!C_{t\rightarrow t+N}^{\scalebox{0.4}{LMPC},j}(x_t^j, N_t^j) =  \notag\\
\min_{\substack{\boldsymbol{M}_t^j,\\ {\bm{\lambda}}^j_{t}, \boldsymbol{g}_t^j } } ~ & \sum_{k=t}^{t+N-1}  h(\bar x^j_{k|t}, \pi^j_{k|t}(\bar x^j_{k|t})) + Q^{j-1}(\bar x^j_{t+N|t}) \label{eq:FTOPC_Cost} \\
\textrm{s.t. }~\!~
&   x^j_{t|t}= \bar x^j_{t|t}=x_t^j,\label{eq:FTOPC_Init} \\
&  \bar x^j_{k+1|t}=A \bar x^j_{k|t}+B \pi^j_{k|t}(\bar x^j_{k|t}), \label{eq:FTOCP_dynNom}\\
&   x^j_{k+1|t}=A x^j_{k|t}+B \pi^j_{k|t}(x^j_{k|t}) + w^j_{k|t}, \label{eq:FTOCP_dyn}\\
&   x^j_{k|t} \in \mathcal{X},~~\pi^j_{k|t}(x^j_{k|t}) \in \mathcal{U},  \label{eq:FTOCP_InpCons}\\
&   x^j_{t+N|t} \in ~\mathcal{CS}^{j-1}, \label{eq:FTOCP_TermCons}\\
&   \pi^{\text{d}}_{k|t}(x_{k|t}) = \textstyle\sum_{s=0}^{k-t-1} M^j_{ks|t}w^j_{s|t} + g^j_{k|t}, \label{eq:FTOCP_DFpolicy} \\
&   \kappa_{k|t}^{j-1}(x^j_{k|t}) =  {\bf{U}}^{j-1} {\bm{\lambda}}^j_{k|t}, \label{eq:FTOCP_DataBasedpolicy1}\\
&   {\bm{\lambda}}^j_{k|t} \in \Lambda^{j-1}(x_{k|t}^j), \label{eq:FTOCP_DataBasedpolicy2} \\
&   \pi^j_{i|t}(x^j_{i|t}) = \pi^{\text{d}}_{i|t}(x_{i|t}),\forall i\in \{t,...,t+N_t^j-1 \}, \label{eq:FTOCP_policy1} \\
&  \pi^j_{i|t}(x^j_{i|t}) = \kappa_{i|t}^{j-1}(x^j_{i|t}), \forall i\in \{t+N_t^j, ..., t\!+\!N\!-\!1\},\label{eq:FTOCP_policy2} \\
&   \forall w^j_{k|t} \in \mathcal{W}, ~ \forall k = \{t, ..., t+N-1\}, \notag
\end{align}
\end{subequations}
where the optimization variables are
\begin{equation*}
    \boldsymbol{M}_t^j = \begin{bmatrix} 0 & \ldots \\
    M_{21|t}& 0 & \ldots \\
    M_{31|t}& M_{32|t} &  \ddots \\
    \vdots
    \end{bmatrix},~\boldsymbol{g}_t^j = \begin{bmatrix} g^j_{t|t}\\
    \vdots \\
    g^j_{t+N-1|t},
    \end{bmatrix}
\end{equation*}
${\bm{\lambda}}^j_{t} = [{\bm{\lambda}}^j_{t|t},\ldots, {\bm{\lambda}}^j_{t+N-1|t}]$. 
Equation~\eqref{eq:FTOCP_policy2} and the parameter $N_t^j$ describe the control policy which defines the evolution of the predicted nominal and uncertain trajectories in \eqref{eq:FTOCP_dynNom}--\eqref{eq:FTOCP_dyn}. In particular, for the first $N_t^j$ predicted time steps the control policy $\pi_{k|t}^j(\cdot)$ equals the disturbance feedback policy~\eqref{eq:FTOCP_DFpolicy}, and for the last $N-N_t^j$ predicted steps $\pi_{k|t}^j(\cdot)$ equals the safe feedback policy~\eqref{eq:FTOCP_DataBasedpolicy1}--\eqref{eq:FTOCP_DataBasedpolicy2}.
Equations \eqref{eq:FTOPC_Init}--\eqref{eq:FTOCP_InpCons} represent input and state constraints which must be satisfied robustly for all disturbance realizations. 
For $k \in \{t, \ldots, N_t^j-1\}$, the robust state and input constraints are reformulated using the  disturbance feedback policy~\eqref{eq:FTOCP_DFpolicy} \cite{goulart2006optimization}. 
On the other hand, from Proposition~\ref{prop:invariance}, we have that for time $k\in\{ N_t^j, \ldots, N \}$ the safe policy~\eqref{eq:FTOCP_DataBasedpolicy1} guarantees robust constraints satisfaction by construction. 
Finally, the terminal constraint \eqref{eq:FTOCP_TermCons} robustly enforces $x_{t+N|t}$ within the robust control invariant set $\mathcal{CS}^{j-1}$. 

The finite time optimal control problem~\eqref{eq:FTOCP} is used to define the LMPC policy from Algorithm~\ref{LMPCpolicy}. Given the measured state $x_t^j$, Algorithm~\ref{LMPCpolicy} solves $N+1$ instances of Problem~\eqref{eq:FTOCP} and it returns the optimal robust $N$-steps policy at time $t$
\begin{equation}\label{eq:optNstepPolicy}
	{\boldsymbol{\pi}_t^{j,*}(\cdot)=[\pi^{j,*}_{t|t}(\cdot), \ldots, \pi^{j,*}_{t+N-1|t}(\cdot)]}
\end{equation} 
and the LMPC cost $J_{t\rightarrow t+N}^{\scalebox{0.4}{LMPC},j}(x_t^j)$. Then, we apply to system~\eqref{eq:sys}
\begin{equation}\label{eq:MPCpolicy}
    u_t^j = \pi^{j,*}_{t|t}(x_t).
\end{equation}
Algorithm~\ref{LMPCpolicy} is resolved at time $t+1$, based on the new state $x_{t+1|t+1} = x_{t+1}^j$, yielding a \textit{moving} or \textit{receding horizon} control strategy. Finally, the control policy~\eqref{eq:optNstepPolicy} is used to compute the robust convex safe set $\mathcal{CS}^j$ as discussed in Section~\ref{sec:LMPCpre}.

\begin{algorithm}[h!] 
	\SetAlgoLined
	\textbf{Input:} System's state $x_t^j$\\
	Set $N_t^{j,*} = \textstyle \text{argmin}_{N_t^j \in \{0, \ldots, N\} } C_{t\rightarrow t+N}^{\scalebox{0.4}{LMPC},j}(x_t^j, N_t^j) $ \\
	Let $\boldsymbol{\pi}_t^{j,*}(\cdot)=[\pi^{j,*}_{t|t}(\cdot), \ldots, \pi^{j,*}_{t+N-1|t}(\cdot)]$ be the optimal solution to problem  $C_{t\rightarrow t+N}^{\scalebox{0.4}{LMPC},j}(x_t^j,N_t^{j,*})$\\
    Set $J_{t\rightarrow t+N}^{\scalebox{0.4}{LMPC},j}(x_t) = \textstyle \text{min}_{N_t^j \in \{0, \ldots, N\} } C_{t\rightarrow t+N}^{\scalebox{0.4}{LMPC},j}(x_t^j, N_t^j)$\\
	\textbf{Output:} $\boldsymbol{\pi}_t^{j,*}(\cdot)$ and $J_{t\rightarrow t+N}^{\scalebox{0.4}{LMPC},j}(x_t^j)$
	\caption{LMPC Algorithm}
	\label{LMPCpolicy}
\end{algorithm}

The above algorithm solves different instances of Problem~\eqref{eq:FTOCP} to synthesize the control policy at iteration $j$. In particular, in line~2 we first solve Problem~\eqref{eq:FTOCP} for all $N_t^j \in \{0,\ldots, N\}$ and then set $N_t^{j,*}$ equal to the parameter which led to the minimum open-loop cost $C_{t\rightarrow t+N}^{\scalebox{0.4}{LMPC},j}(x_t^j, N_t^{j,*})$. Afterwards, in lines $3$-$4$ we set the optimal control policy $\boldsymbol{\pi}_t^{j,*}(\cdot)$ and the open-loop cost $J_{t\rightarrow t+N}^{\scalebox{0.4}{LMPC},j}(x_t)$ equal to the optimal policy and the open-loop cost of problem $C_{t\rightarrow t+N}^{\scalebox{0.4}{LMPC},j}(x_t^j, N_t^{j,*})$, respectively.

\vspace{0.2cm} \begin{remark}
The optimal control problem~\eqref{eq:FTOCP} computes a sequence of policies ${\boldsymbol{\pi}_t^{j,*}(\cdot)=[\pi^{j,*}_{t|t}(\cdot), \ldots, \pi^{j,*}_{t+N-1|t}(\cdot)]}$, which robustly steer system~\eqref{eq:sys} from $x_t^j$ to the robust convex safe set $\mathcal{CS}^{j-1}$. 
Notice that if at iteration $j-1$ more data is collected, then the robust convex safe set $\mathcal{CS}^{j-1}$ from~\eqref{eq:CS} may be enlarged. As a result, when the tasks are defined over longer time periods $T^j$, the proposed iterative control design procure may synthesize control policies with larger regions of attraction.
\end{remark}

\subsection{Design Choices}\label{sec:discussion}
In standard robust MPC at each time step we solve an optimal control problem over a fixed space of feedback policies. On the other hand, in Problem~\eqref{eq:FTOCP} the space of feedback policies changes as a function of the predicted time step $k$. In particular, the predicted trajectory is computed using a disturbance feedback policy for $k \leq N_t^j$ and a safe feedback policy~\eqref{eq:FTOCP_DataBasedpolicy1}--\eqref{eq:FTOCP_DataBasedpolicy2} for $k>N_t^j$. In the following we discuss why this strategy allows us to guarantee recursive constraint satisfaction.

Recall that in predictive control recursive constraint satisfaction is ensured using a terminal constraint set. In particular, the terminal constraint set should be (robust) control invariant, for a feedback policy that can be used by the (robust) MPC to forecast the evolution of the system~\cite{BorrelliBemporadMorari_book}. 
Notice that a disturbance feedback policy (or equivalently an affine state feedback policy~\cite{goulart2006optimization}) may not be able to robustly constrain the evolution of the system within the terminal constraint set $\mathcal{CS}^j$. 
For this reason, in Problem~\eqref{eq:FTOCP} we used a time-varying feedback policy, which is defined by the parameter $N_t^j$, and in Algorithm~\ref{LMPCpolicy} we solved Problem~\eqref{eq:FTOCP} for different values of $N_t^j$. This strategy guarantees that the safe policy can be used to robustly constrain the evolution of the predicted system within the robust safe set $\mathcal{CS}^j$, and it is used in Theorem~\ref{th:recFeas} to show that the LMPC~\eqref{eq:FTOCP} and~\eqref{eq:MPCpolicy} guarantees recursive feasibility.

Furthermore, we comment on the computational tractability of the proposed strategy. As already mentioned, Algorithm~\ref{LMPCpolicy} solves $N+1$ instances of Problem~\eqref{eq:FTOCP} to forecast the evolution of the system using either the disturbance feedback policy or the safe policy from Section~\ref{sec:dataBasedPolicy}. We underline that these $N+1$ optimal control problems are independent and can be solved in parallel. Therefore, when parallel computing is available, the online computational complexity of the proposed strategy is independent on the controller horizon.

Finally, we emphasize the differences between the linear control policy defined by the feedback gain $K$ from Assumption~\ref{ass:O_inf},  the safe control policy $\kappa^{j,*}(\cdot)$ in~\eqref{eq:theSafePolicy}, and the LMPC policy~\eqref{eq:MPCpolicy}. We notice that the linear feedback policy $\pi^{\mathrm{linear}}(x) = K x$ may be used to robustly constrain the evolution of the system within the robust positive invariant set $\mathcal{O}$, which may be a small neighborhood of the origin. On the other hand, the safe control policy $\kappa^{j,*}(\cdot)$ is defined on the robust convex safe set $\mathcal{CS}^j \subseteq \mathcal{O}$ and it may be used to steer the system from any state $x \in \mathcal{CS}^j$ to the neighborhood of the origin $\mathcal{O}$, as shown in the result section. Finally, the LMPC policy~\eqref{eq:MPCpolicy} is defined on a larger domain $\mathcal{C}^j \subseteq \mathcal{CS}^j$ and it may be used to enlarge the robust convex safe set by collecting new data, when the initial condition $x_0^j$ is selected on the boundary of $\mathcal{C}^j$. A strategy to select such initial condition is described in Section~\ref{sec:regAtt} and its efficacy is demonstrated in the result Section~\ref{sec:resIterativeEnlargement}.

\section{Properties}\label{sec:properties}
This section shows that the LMPC policy~\eqref{eq:MPCpolicy} satisfies our design requirements from Section~\ref{sec:Obj}.

\subsection{Recursive Feasibility}
We show that if Problem~\eqref{eq:FTOCP} is feasible at time $t=0$ for some $N_0^j \in \{0, \ldots, N\}$, then the LMPC policy~\eqref{eq:MPCpolicy} guarantees that state and input constraints are recursively satisfied. 
First, we define the set
\begin{equation}\label{eq:feasRegion}
    \mathcal{C}^j = \{ x \in \mathbb{R}^n : \exists N_0^j \in \{0, \ldots, N\}, C_{0\rightarrow N}^{\scalebox{0.5}{LMPC},j}(x, N_0^j) < \infty \}
\end{equation}
collecting the states from which Problem~\eqref{eq:FTOCP} is feasible for some $N_0^j \in \{0, \ldots, N\}$. Next, we show that if $x_0^j \in \mathcal{C}^j$, then Problem~\eqref{eq:FTOCP} is feasible for all time $t\geq 0$ and for some $N_t^j \in \{0, \ldots, N\}$.

\begin{theorem}\label{th:recFeas}
Consider the closed-loop system~\eqref{eq:sys} and \eqref{eq:MPCpolicy}. Let Assumptions~\ref{ass:first}--\ref{ass:cost} hold and $x_0^j \in \mathcal{C}^j$. Then, for all time $t\geq0$ Problem~\eqref{eq:FTOCP} is feasible for some $N_t^j \in \{0, \ldots, N\}$, and the closed-loop system~\eqref{eq:sys} and \eqref{eq:MPCpolicy} satisfies state and input constraints.
\end{theorem}

\begin{proof}
Assume that at time $t$ Problem~\eqref{eq:FTOCP} is feasible for some $N_t^j \in \{0, \ldots, N\}$. At the next time $t+1$, by Proposition~\ref{prop:invariance} we know that, for $\kappa^{j-1,*}(\cdot) \in \mathcal{SP}^{j-1}$, the following candidate policy 
\begin{equation}\label{eq:candidateSolution}
    [\pi^{j,*}_{t+1|t}(\cdot), \ldots, \pi^{j,*}_{t+N-1|t}(\cdot), \kappa^{j-1,*}(\cdot)]    
\end{equation}
is feasible for the Problem~\eqref{eq:FTOCP} for some $N_{t+1}^j \in \{0, \ldots, N\}$. \\
By assumption $x_0^j \in \mathcal{C}^j$, thus we have that Problem~\eqref{eq:FTOCP} is feasible at time $t=0$ for some $N_0^j \in \{0, \ldots, N\}$. Concluding, we have shown that if Problem~\eqref{eq:FTOCP} is feasible for some $N_t^j \in \{0, \ldots, N\}$ at time $t$, then Problem~\eqref{eq:FTOCP} is feasible for some $N_{t+1}^j \in \{0, \ldots, N\}$ at $t+1$. Therefore by induction we have that for all time $t\geq0$ Problem~\eqref{eq:FTOCP} is feasible for some $N_t^j \in \{0, \ldots, N\}$ and the closed-loop system~\eqref{eq:sys} and~\eqref{eq:MPCpolicy} satisfies state and input constraints~\eqref{eq:constraintSet}.
\end{proof}

\subsection{Input to State Stability (ISS)}
In this section, we show that the closed-loop system~\eqref{eq:sys} and~\eqref{eq:MPCpolicy} is ISS with respect to $\mathcal{O}$. 
We recall that in standard MPC strategies the finite time optimal control problem can be reformulated as a parametric Quadratic Program (QP). This fact is used in~\cite{goulart2006optimization} to show continuity of the open-loop cost and then to prove ISS of the origin. 
In the proposed approach, the open-loop cost from Algorithm~\ref{LMPCpolicy} 
\begin{equation*}
    J_{t\rightarrow t+N}^{\scalebox{0.4}{LMPC},j}(x_t) = \textstyle \text{min}_{N_t^j \in \{0, \ldots, N\} } C_{t\rightarrow t+N}^{\scalebox{0.4}{LMPC},j}(x_t^j, N_t^j)
\end{equation*}
is not given by the solution to a parametric QP. 
Therefore, continuity cannot be guaranteed, and the standard technique from~\cite{goulart2006optimization} cannot be used to prove ISS.
Instead, we introduce the standard definition of dissipative-form ISS-Lyapunov function for the robust invariant set $\mathcal{O}$  \cite{lin1995various,grune2014iss} and we show that the cost of the LMPC $J_{t\rightarrow t+N}^{\scalebox{0.4}{LMPC},j}(x_t)$ is an ISS-Lyapunov function.

\begin{definition}
A dissipative-form ISS-Lyapunov function for the closed-loop system \eqref{eq:sys} and~\eqref{eq:MPCpolicy} and the invariant set $\mathcal{O}$ is a function $V:\mathbb{R}^n \rightarrow \mathbb{R}_{\geq0}$ such that there exists $\alpha_1, \alpha_2, \alpha \in \mathcal{K}_\infty$, and $\sigma \in \mathcal{K}$ so that for all $x_t \in \mathbb{R}^n$ and $w_t \in \mathbb{R}^m$,
\begin{subequations}
\begin{align}
    &\alpha_1(|x_t|_\mathcal{O}) \leq V(x_t) \leq \alpha_2(|x_t|_\mathcal{O}) \label{eq:zeroProperty}\\
    &V(Ax_t+B\pi^{j,*}_{t|t}(x_t)+w)-V(x_t) \leq -\alpha(|x_t|_\mathcal{O}) + \sigma(||w_t||).
\end{align}
\end{subequations}
\end{definition}

Notice that, as in \cite{grune2014iss}, no assumptions on the continuity of $V(\cdot)$ are required. However \eqref{eq:zeroProperty} implies that $V(\cdot)$ is continuous on the boundary of $\mathcal{O}$. The above definition can be used to show that the closed-loop system \eqref{eq:sys} and~\eqref{eq:MPCpolicy} is ISS with respect to the invariant set $\mathcal{O}$, as described by the following proposition.

\begin{proposition}\label{prop:ISS}
The following statements are equivalent:
\begin{itemize}
    \item The closed-loop system \eqref{eq:sys} and~\eqref{eq:MPCpolicy} is ISS with respect to the robust invariant set $\mathcal{O}$.
    \item There exists a dissipative-form ISS-Lyapunov function $V(\cdot)$.
\end{itemize}
\end{proposition}
\begin{proof}
The proof follows from \cite[Theorem~2.3]{grune2014iss} substituting $|x|$ with $|x|_\mathcal{O}$. Note that we can replace $|x|$ with $|x|_\mathcal{O}$ as by \eqref{eq:zeroProperty} we have that $V(x)=0$ iff $|x|_\mathcal{O}=0$.
\end{proof}
\begin{proposition}\label{prop:costBound}
Let Assumptions~\ref{ass:first}--\ref{ass:cost} hold, $x_t^j \in \mathcal{C}^j$ and define the closed-loop system dynamics 
\begin{equation*}
    f^j_t(x_t^j,w_t^j)=Ax_t+B \pi^{j,*}_{t|t}(x_t^j)+w_t^j,
\end{equation*}
where $\pi^{j,*}_{t|t}(\cdot)$ is the optimal policy~\eqref{eq:optNstepPolicy}. Then there exists a constant $L > 0$ such that
\begin{equation*}
\begin{aligned}
    J^{\scalebox{0.5}{LMPC},j}_{t+1 \rightarrow t+1+N}(f_t^j(x_t^j,w_t^j) ) - & J^{\scalebox{0.5}{LMPC},j}_{t \rightarrow t + N}(x_t^j )\\
    &  \leq -\alpha(|x_t^j|_\mathcal{O}) + L ||w_t^j||,
\end{aligned}
\end{equation*}
$\forall t \geq 0$, $\forall w_t^j \in \mathcal{W}$ and $\alpha \in \mathcal{K}_\infty$.
\end{proposition}

\begin{proof}
The proof can be found in the Appendix.
\end{proof}

The above propositions allow us to prove that the closed-loop system~\eqref{eq:sys} and~\eqref{eq:MPCpolicy} is ISS with respect to $\mathcal{O}$.\\

\begin{theorem}\label{th:conv}
Consider the closed-loop system~\eqref{eq:sys} and~\eqref{eq:MPCpolicy}. Let Assumptions~\ref{ass:first}--\ref{ass:cost} hold and assume that $x_0^j \in \mathcal{C}^j$, then the closed-loop system~\eqref{eq:sys} and~\eqref{eq:MPCpolicy} is Input-to-State Stable (ISS) for the robust positive invariant set $\mathcal{O}$.
\end{theorem}
\begin{proof} First we show that the set $\mathcal{O}$ is robust positive invariant for the closed-loop system~\eqref{eq:sys} and~\eqref{eq:MPCpolicy}. Assume that at time $t$ of iteration $j$ the state $x_t^j \in \mathcal{O}$ and recall that the disturbance feedback policy \eqref{eq:FTOCP_DFpolicy} is equivalent to state a feedback policy \cite{goulart2006optimization}, then we have that the candidate policy
\begin{equation*}
    [\pi^{j,*}_{t|t}(x)=Kx,\ldots,\pi^{j,*}_{t+N-1|t}(x)=Kx]
\end{equation*}
is feasible at time $t$ of the $j$th iteration for $N_t^j = N$. Now we notice that the cost associated with the above feasible policy is zero. Therefore, we have that $u_t^j=\pi^{j,*}_{t|t}(x_{t}^{j})=Kx_t^j$,  which together with Assumption~\ref{ass:O_inf} implies that the closed-loop system $x_{t+1}^j = Ax_t^j + B\pi^{j,*}_{t|t}(x_{t}^{j}) +w_t^j \in \mathcal{O}, \forall w_t^j \in \mathcal{W}$ and that $\mathcal{O}$ is robust positive invariant for the closed-loop system~\eqref{eq:sys} and~\eqref{eq:MPCpolicy}. Next, notice that for a fixed $N_t^j$ the function $C_{t\rightarrow t+N}^{\scalebox{0.4}{LMPC},j}(\cdot, N_t^j):\mathbb{R}^n\rightarrow\mathbb{R}$ is continuous at the boundary of the set $\mathcal{O}$. 
Therefore, the function $J^{\scalebox{0.5}{LMPC},j}_{t \rightarrow t+  N}( \cdot )$, which is defined as the point-wise minimum of the functions $C_{t\rightarrow t+N}^{\scalebox{0.4}{LMPC},j}(\cdot, N_t^j)$ for $N_t^j \in \{0,\ldots,N\}$, is continuous at the boundary of $\mathcal{O}$.
Continuity at the boundary of $\mathcal{O}$, Assumption~\ref{ass:cost} and \eqref{eq:FTOPC_Cost} imply the existence of $\alpha_1, \alpha_2 \in \mathcal{K}_\infty$ such that  $\forall x_t \in \mathcal{C}^j$
\begin{equation}
    \alpha_1(|x_t|_\mathcal{O}) \leq h(x_t, 0) \leq J^{\scalebox{0.5}{LMPC},j}_{t \rightarrow  t+N}( x_t ) \leq \alpha_2(|x_t|_\mathcal{O}).
\end{equation}
In the above equation, we used compactness of the sets $\mathcal{X}$ and $\mathcal{U}$.
Finally, from Proposition~\ref{prop:costBound}, we have that $\forall x_t \in \mathcal{C}^j$
\begin{equation*}
\begin{aligned}
     J^{\scalebox{0.5}{LMPC},j}_{t+1 \rightarrow t+1+N}( Ax_t+B\pi^{j,*}_{t|t}(x_t)&+w_t )-J^{\scalebox{0.5}{LMPC},j}_{t \rightarrow t+N}( x_t ) \\
     &~\leq - \alpha(|x_t|_\mathcal{O})+\sigma(||w_t||_2)
\end{aligned}
\end{equation*}
and therefore $J^{\scalebox{0.5}{LMPC},j}_{t \rightarrow t+N}( \cdot)$ is a ISS-Lyapunov function and the closed-loop system~\eqref{eq:sys} and~\eqref{eq:MPCpolicy} is Input-to-State Stable for the robust positive invariant set $\mathcal{O}$. 
\end{proof}

\vspace{0.1cm}
\vspace{0.2cm} \begin{remark}
We underline that the ISS property does not imply that the closed-loop system converges in finite time to the set $\mathcal{O}$. However, finite time convergence may be obtained following the approach presented in~\cite{langson2004robust}, if at time $T^j$ the system has not reach $\mathcal{O}$. In particular, the optimal policy ${\boldsymbol{\pi}_{T^j}^{j,*}(\cdot)=[\pi^{j,*}_{{T^j}|t}(\cdot), \ldots, \pi^{j,*}_{{T^j}+N-1|t}(\cdot)]}$ computed by the finite time optimal control problem~\eqref{eq:FTOCP} can be used to robustly steer the system to the terminal set~$\mathcal{CS}^{j-1}$. Then, from a state $x_{{T^j}+N-1}^j\in\mathcal{CS}^j$ it would be possible to solve the finite time optimal control problem~\eqref{eq:FTOCP} for $j-1$ and leverage the time-varying optimal policy ${\boldsymbol{\pi}_{T^j+N}^{j-1,*}(\cdot)}$ to steer the system in finite time to $\mathcal{CS}^{j-2}$. This procedure may be iterated to steer the system in finite time to $\mathcal{CS}^{0} = \mathcal{O}$. Finally, we point out that robust exponential stability of the closed-loop system may be obtained decoupling the nominal and error dynamics as shown in~\cite{mayne2005robust}. However, this strategy cannot be applied when using a disturbance feedback policy as in~\eqref{eq:FTOCP}.

\end{remark}

\subsection{Performance Bound}
Finally, we show that whenever $x_0^j\in \mathcal{CS}^{j-1}$ the robust $Q$-function at iteration $j-1$ is an upper-bound to the closed-loop cost of the certainty equivalent system at the next $j$th iteration.
\begin{theorem}\label{th:cost}
Consider the certainty equivalent system \eqref{eq:disturbanceFreeSys} in closed-loop with the LMPC \eqref{eq:FTOCP} and~\eqref{eq:MPCpolicy}. Let Assumptions~\ref{ass:first}--\ref{ass:cost} hold and $x_0^j = \bar x_0^j \in \mathcal{CS}^{j-1} \subseteq \mathcal{C}^j$. Then, we have that the iteration cost of the certainty equivalent closed-loop system \eqref{eq:disturbanceFreeSys} and~\eqref{eq:MPCpolicy} is an upper-bound to the $Q$-function constructed at the previous iteration, i.e.,
\begin{equation}
        J^{j}_{0\rightarrow  T^j}(\bar x_0^j) = \sum_{t=0}^{T^j} h\big(\bar x_t^j, \pi_{t|t}^{j,*}(\bar x_t^j)\big) \leq Q^{j-1}(\bar x_0^j)
\end{equation}
where $\pi_{t|t}^{j,*}(\cdot)$ is the optimal control policy in~\eqref{eq:optNstepPolicy}.
\end{theorem}

\begin{proof}
By Proposition~\ref{prop:QfunDecrease} we have that
\begin{equation}\label{eq:boundQfun}
\begin{aligned}
     Q^{j-1}(\bar x_0^j) & \geq h(\bar x_0^j, \kappa^{j-1,*}(\bar x_0^j) ) + Q^{j-1}(\bar x_1^j) \\
     & \geq h(\bar x_0^j, \kappa^{j-1,*}(\bar x_0^j) ) + h(\bar x_1^j, \kappa^{j-1,*}(\bar x_1^j) ) \\
     & \quad \quad \quad \quad \quad \quad \quad \quad \quad \quad \quad \quad \quad + Q^{j-1}(\bar x_2^j) \\
     & \geq \sum_{k=0}^{N-1} h(\bar x_k^j, \kappa^{j-1,*}(\bar x_k^j) ) + Q^{j-1}(\bar x_{N}^j) \\
     & \geq J^{\scalebox{0.5}{LMPC},j}_{0\rightarrow N}(\bar x_0^j),
\end{aligned}
\end{equation}
where the last inequality holds by the feasibility of the safe policy from Section~\ref{sec:dataBasedPolicy} for $x_0^j \in \mathcal{CS}^{j-1} \subseteq \mathcal{C}^j$.\\
Now consider the LMPC cost at time $t$, by Proposition~\ref{prop:QfunDecrease} we have that
\begin{equation*}
\begin{aligned}
    &J^{\scalebox{0.5}{LMPC},j}_{t \rightarrow t + N}(\bar x_t^j) =  \sum_{k = t}^{t+N-1} h(\bar x_{k|t}^{j,*}, \pi_{k|t}^{j,*}(\bar x_{k|t}^{j,*})) + Q^{j-1}( \bar x_{t+N|t}^{j,*} ) \\
     &~~\geq  \sum_{k = t}^{t+N-1} h(\bar x_{k|t}^{j,*}, \pi_{k|t}^{j,*}(\bar x_{k|t}^{j,*})) + h\Big(\bar x_{t+N|t}^{j,*} , \kappa^{j-1,*}\big( \bar x_{t+N|t}^{j,*} \big)\Big) \\
     &~~~~~~~~~~~~~~~+ Q^{j-1}\Big(A\bar x_{t+N|t}^{j,*} + B \kappa^{j-1,*}\big( \bar x_{t+N|t}^{j,*} \big)\Big) \\
   & ~~~=  h(\bar x_{t|t}^{j,*}, \pi_{t|t}^{j,*}(\bar x_{t|t}^{j,*})) + \sum_{k = t+1}^{t+N-1} h(\bar x_{k|t}^{j,*}, \pi_{k|t}^{j,*}(\bar x_{k|t}^{j,*})) \\
     & ~~~ + h\Big(\bar x_{t+N|t}^{j,*} , \kappa^{j-1,*}\big( \bar x_{t+N|t}^{j,*} \big)\Big)    \\
     &~~~~~~~~~~~~~~~~~~+ Q^{j-1}\Big(A\bar x_{t+N|t}^{j,*} + B \kappa^{j-1,*}\big( \bar x_{t+N|t}^{j,*} \big)\Big) \\
      & ~~~ \geq h(\bar x_{t}^{j}, \bar u_{t}^{j}) + J^{\scalebox{0.5}{LMPC},j}_{t+1 \rightarrow t+1 + N}( A \bar x_t^{j} + B \bar u_{t}^{j} ),
\end{aligned}
\end{equation*}
where $\bar u_{t}^{j} =  \pi_{t|t}^{j,*}(\bar x_{t}^{j})$.
The above equation implies that the LMPC cost is non-increasing and for the closed-loop trajectory of the certainty equivalent we have
\begin{equation*}\label{eq:costDecrease}
  J^{\scalebox{0.5}{LMPC},j}_{t+1 \rightarrow t+1 + N}( \bar x_{t+1}^{j} ) - J^{\scalebox{0.5}{LMPC},j}_{t \rightarrow t + N}(\bar x_t^j) \leq -h(\bar x_{t}^{j}, \bar u_{t}^{j}),
\end{equation*}
which in turns implies that
\begin{equation*}
\begin{aligned}
    J^{\scalebox{0.5}{LMPC},j}_{0 \rightarrow N}(\bar x_0^j) & \geq \sum_{k=0}^{T^j} h(\bar x_{k}^{j}, \bar u_{k}^{j}) + J^{\scalebox{0.5}{LMPC},j}_{T^j+1 \rightarrow T^j+1 + N}( \bar x_{T^j+1}^{j} ) \\
    &\geq \sum_{k=0}^{T^j} h(\bar x_{k}^{j}, \bar u_{k}^{j})=J^{j}_{0\rightarrow  T^j}(\bar x_0^j).
\end{aligned}
\end{equation*}
Finally, from the above equation and~\eqref{eq:boundQfun} we conclude that 
\begin{equation*}
    Q^{j-1}(\bar x_0^j) \geq J^{\scalebox{0.5}{LMPC},j}_{0 \rightarrow N}(\bar x_0^j) \geq J^{j}_{0\rightarrow  T^j}(\bar x_0^j).
\end{equation*}
\end{proof}

\section{Region of Attraction Approximation}\label{sec:regAtt}
In this section, we present an algorithm which may be used to select the initial condition $x_0^j$ or to approximate the region of attraction of the LMPC policy $\mathcal{C}^j$ in \eqref{eq:feasRegion}. We underline that the region of attraction $\mathcal{C}^j$ may be computed solving problem~\eqref{eq:FTOCP} as a set of parametric optimization problems with parameter $x_t^j$~\cite{BorrelliBemporadMorari_book}. However, this computation may be prohibitive and therefore we propose a strategy to compute an inner approximation to the region of attraction $\mathcal{C}^j$.

Given a vector $d\in \mathbb{R}^n$, we define the following optimization problem:
\begin{equation}\label{eq:regAttApp}
\begin{aligned}
P(N_t^j, d) =
\min_{\substack{x_0, \boldsymbol{M}_t^j,\\ {\bm{\lambda}}^j_{t}, \boldsymbol{g}_t^j } } & ~~ d^\top  x_0  \\
\textrm{s.t. }& ~~ (d^\perp )^\top x_0 =0 \\
& ~~ x^j_{t|t}= \bar x^j_{t|t}=x_0\\
& ~~ \eqref{eq:FTOCP_dynNom}-\eqref{eq:FTOCP_policy2}\\
& ~~   \forall w^j_{k|t} \in \mathcal{W}, \forall k = \{t, ..., t\!+\!N\!\!-\!\!1\},
\end{aligned}
\end{equation}
where $d^\perp \in \mathbb{R}^n$ is a vector perpendicular to $d \in \mathbb{R}^n$. Basically, the above optimization problem finds the optimal initial state $x_0^*$, which is farthest from the origin along the direction $d$ and guarantees that problem $C_{t\rightarrow t+N}^{\scalebox{0.5}{LMPC},j}(x_0, N_t^j)$ from~\eqref{eq:FTOCP} is feasible. Therefore, given a user-defined set of vectors $\mathcal{D}=\{d^1, \ldots, d^k\}$, problem~\eqref{eq:regAttApp} can be solved repeatedly to approximate the region of attraction~${\mathcal{C}}^j$. In particular, for each vector $d^i \in \mathcal{D}$ and parameter $N_t^j \in \{1,\ldots,N\}$ we solve problem $P(N_t^j, d^i)$. Afterwards, we define the approximated region of attraction $\mathcal{\tilde{C}}^j$ as the convex hull of the optimal initial states $x_0^*$ associated with different $N_t^j$ and vector $d^i$. Algorithm~\ref{RegAttrAppr} summaries this procedure.

\begin{algorithm}[h!] 
	\SetAlgoLined
	\textbf{Input:} Set of vectors $\mathcal{D}=\{d^1, \ldots, d^k\}$ and horizon $N$\\
	Initialize $\mathcal{\tilde{C}}^j = \varnothing$ \\
	\For{$d^i \in \mathcal{D}$}{
	\For{$N_t^j \in \{1,\ldots,N\}$}{
	Solve $P(N_t^j, d^i)$ from~\eqref{eq:regAttApp}\\
	Let $x_0^*$ be the optimal initial state from $P(N_t^j, d^i)$\\
	Set $\mathcal{\tilde{C}}^j = \text{Conv}\{\mathcal{\tilde{C}}^j\cup \{x_0^*\}\}$
	}}
	\textbf{Output:} Approximate region of attraction $\mathcal{\tilde{C}}^j$
	\caption{Region of Attraction Approximation}
	\label{RegAttrAppr}
\end{algorithm}

Finally, we underline that the optimization problem~\eqref{eq:regAttApp} may be used to select the initial condition at iteration $j$. In particular given a direction $d^j$, the initial condition $x_0^j$ may be chosen as the optimal initial state $x_0^*$ for problem $P(N, d^j)$, as shown in the result section.

\section{Simulation Results}\label{sec:results}
We test the proposed controller on a system subject to bounded additive uncertainty. First, we show that the proposed strategy is able to improve the performance of a system executing an iterative task. Afterwards, we show that the proposed LMPC can be used to iteratively construct a robust convex safe set $\mathcal{CS}^j$, which is defined over progressively larger regions of the state space. Finally, we show that data collected by the LMPC can be exploited to construct the safe policy~\eqref{eq:theSafePolicy}, which robustly steers the uncertain system from any state within the robust safe set $\mathcal{CS}^j$ to the set $\mathcal{O}$.

We consider the following double integrator system
\begin{equation*}
    x_{t+1} = \begin{bmatrix} 1 & 1 \\ 0 & 1 \end{bmatrix} x_t + \begin{bmatrix} 0 \\ 1 \end{bmatrix} u_t + w_t,
\end{equation*}
where $w_t \in \{w \in \mathbb{R}^2: ||w||_\infty \leq 0.1 \}$, subject to the following constraints $x_t \in \mathcal{X} = \{x\in \mathbb{R}^2 : ||x||_\infty \leq 10 \}$
and $u_t \in \mathcal{U} = \{ u \in \mathbb{R}: ||u||_\infty \leq 1 \}$ for all time instant $t \geq 0$. Furthermore, we define the running cost $h(x,u) = 10|x|_\mathcal{O} + |u|_{K\mathcal{O}}$, which satisfies Assumption~\ref{ass:cost}. 

\subsection{Iterative Task}
We use the LMPC~\eqref{eq:FTOCP} and~\eqref{eq:MPCpolicy} to iteratively steer the system from $x_0 = [5.656; 0]$ to the robust invariant set $\mathcal{O}$. We designed a sub-optimal robust MPC to perform the $0$th iteration and to construct the robust safe set $\mathcal{CS}^0$ and the robust $Q$-function $Q^0(\cdot)$, which are used to initialize the LMPC with $N=3$. We underline that the closed-loop performance of the robust MPC can be improved tuning its parameters. The goal of this section is to show that, given a safe sub-optimal policy, the proposed strategy may be used to iteratively improve the closed-loop performance. In the next section, we will show that the proposed strategy allows us also to enlarge the region of attraction associated with the safe sub-optimal policy used to initialize the algorithm.

\begin{figure}[t!]
\centering
\includegraphics[width=\columnwidth]{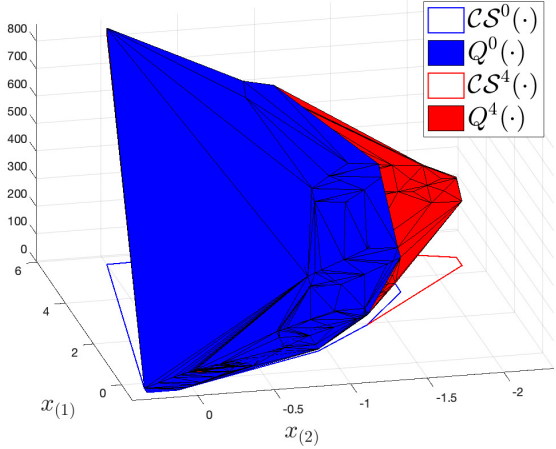}
\caption{Comparison between the robust safe set and $Q$-function at the first and last iteration.}
\label{fig:iterativeQfunEvolution}
\end{figure}

We perform $4$ iterations of the control task for the certainty equivalent system. At each $j$th iteration, we store the LMPC predicted policy and the closed-loop data in order to construct the robust safe set $\mathcal{CS}^j$ and the robust $Q$-function $Q^j(\cdot)$. Table~\ref{table:closedLoopCost} shows that the closed-loop cost of the certainty equivalent system decreases, until it converges to a steady state value after $4$ iterations.

\begin{table}[h!]
\caption{Closed-loop cost $J^j_{0\rightarrow T^j}(x_0)$ for iteration $i\in \{0,\ldots,4\}$.}
\label{table:closedLoopCost}
\centering
\begin{tabular}{cccccc} \midrule
$i=0$& $i=1$& $i=2$&$i=3$&$i=4$ \\
$863.4245$        & $827.9588$         & $827.9380$        & $827.9371$ & $827.9371$\\ \midrule
\end{tabular}
\end{table}

Finally, in Figure~\ref{fig:iterativeQfunEvolution} we compare the robust safe set and the robust $Q$-function at the first and last iteration. First, we notice that the robust safe set, which represents the domain of the $Q$-function, is enlarged. Furthermore, we confirm that $Q^j(\cdot)$ is non-increasing (i.e. $Q^0(x) \leq Q^{4}(x), \forall x \in \mathcal{CS}^{4}$) and therefore it guarantees better bounds on the performance of the certainty equivalent closed-loop system~\eqref{eq:disturbanceFreeSys} and~\eqref{eq:MPCpolicy}, as shown in Theorem~\ref{th:cost}.

\begin{figure}[t]
\centering
\includegraphics[width=\columnwidth]{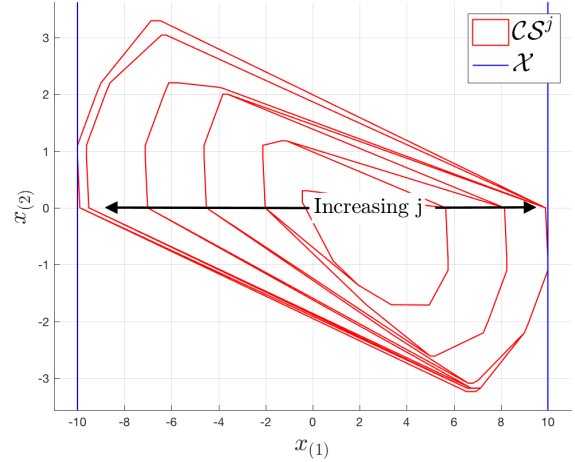}
\caption{Evolution of the robust safe set at each iteration.}
\label{fig:evolutionRS}
\end{figure} 

\begin{figure}[b]
\centering
    \includegraphics[width=\columnwidth]{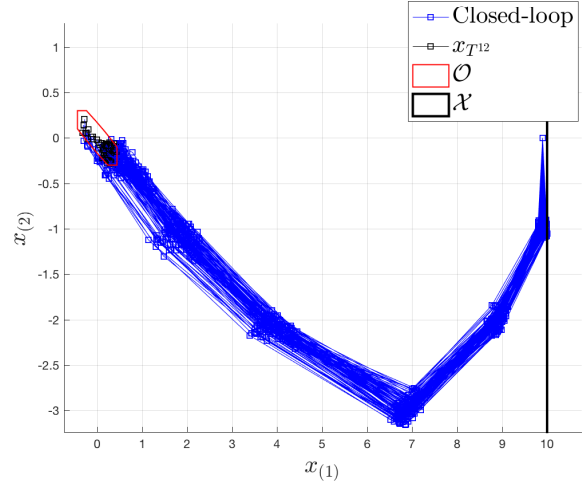}
\caption{Closed-loop trajectories for different disturbance realizations.}
\label{fig:closedloop}
\end{figure}

\subsection{Domain Enlargement}\label{sec:resIterativeEnlargement}
We show that the domain of the LMPC policy may be iteratively enlarged. 
At each iteration, we simulate the uncertain closed-loop system~\eqref{eq:sys} and~\eqref{eq:MPCpolicy} and we store both the closed-loop data and the predicted LMPC policy~\eqref{eq:MPCpolicy}. The stored data are used to construct the robust safe set $\mathcal{CS}^j$ and robust $Q$-function as described in Section~\ref{sec:LMPCpre}.

\begin{figure*}[h!]
\includegraphics[width=1.0\textwidth]{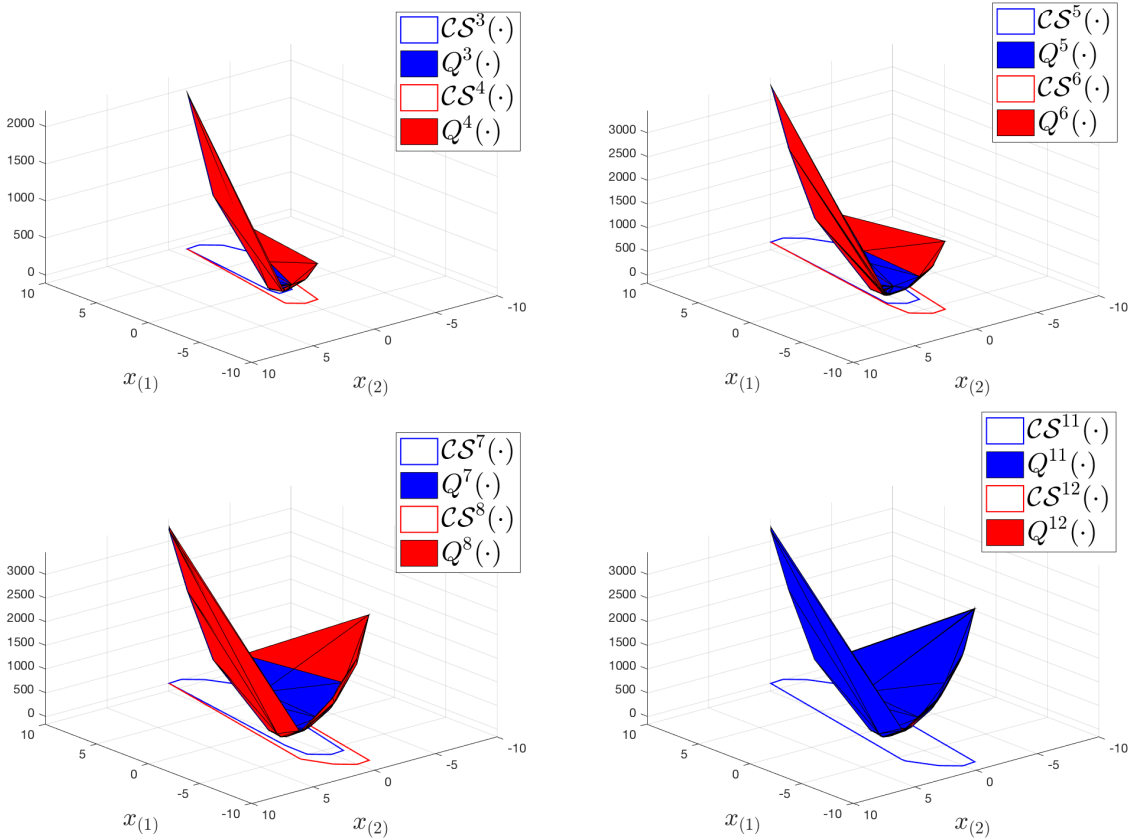}
\caption{Evolution of the robust $Q$-function $Q^j$ at each iteration. Notice that $Q^j(\cdot)$ (in blue) is a lower-bound to $Q^{j+1}(\cdot)$ (in red) for all $i \in \{3,5,7,11\}$, until convergence is reached and $Q^{11}(\cdot) =Q^{12}(\cdot)$.}
\label{fig:evolutionQ}
\end{figure*} 

Furthermore, we use the strategy proposed in Section~\ref{sec:regAtt} to iteratively select the initial condition. At each iteration $j$, we define the vector $d^j =[(-1)^j,0]$ and we pick $x_0^j$ as the optimal initial state $x_0^*$ from problem $P(N, d^j)$ in~\eqref{eq:regAttApp}.
Notice that by definition~\eqref{eq:CS} we have that $\mathcal{CS}^j \subseteq \mathcal{CS}^{j+1}$, therefore the set of states which can be steered to $\mathcal{O}$ by the LMPC~\eqref{eq:FTOCP} and~\eqref{eq:MPCpolicy} does not shrink, i.e., $\mathcal{C}^j \subseteq \mathcal{C}^{j+1}$. 

Figure~\ref{fig:evolutionRS} shows that the robust convex safe set $\mathcal{CS}^j$ grows at each iteration, until it converges to a set which saturates the state constraints. We underline that the closed-loop data used to enlarge $\mathcal{CS}^j$ are generated by the LMPC, which steers the system to regions of the state space associated with low cost values. In other words, the growth of the robust safe set is cost-driven. More importantly, the iterative enlargement of the $\mathcal{CS}^j$ is performed safely. Indeed, the LMPC guarantees robust state and input constraints satisfaction at each iteration.

Figure~\ref{fig:closedloop} shows $100$ Monte Carlo simulations of the closed-loop system for the $12$th iteration. We notice that the closed-loop trajectories satisfy state constraints and converge to the set $\mathcal{O}$, regardless of the disturbance realization.

\begin{table}[b!]
\centering
\caption{Performance of the LMPC policy~\eqref{eq:MPCpolicy} and the safe policy~\eqref{eq:theSafePolicy} in closed-loop with the uncertain system~\eqref{eq:sys}.}\label{tab:comparison}
\begin{tabular}{c|cc}
            & Average solver time & Average closed-loop cost \\\midrule
LMPC        &        $4.6$s   &         $77.1$             \\
Safe Policy &       $5$ms  &         $80$         \\\midrule
\end{tabular}
\end{table}

Finally, Figure~\ref{fig:evolutionQ} shows the growth of the $Q$-function $Q^j(\cdot)$, which is non-increasing at each iteration. It is important to underline that $Q^j(\cdot)$ is piece-wise affine as it is the solution to a parametric LP \cite{BorrelliBemporadMorari_book}. Furthermore, we notice that $Q^j(\cdot)$, which is an upper-bound of the closed-loop cost of the disturbance-fee system, resembles a quadratic function. This result is expected as the optimal value function for this problem is piece-wise quadratic \cite{BorrelliBemporadMorari_book}.

\subsection{Exploiting the safe policy}
In this section, we use the stored data from Section~\ref{sec:dataBasedPolicy} to construct the safe policy~\eqref{eq:theSafePolicy}. We tested this policy for $1000$ Monte Carlo simulations, where we randomly sampled the initial condition $x_0$ from the robust convex safe set $\mathcal{CS}^{12}$. We confirm that, for all initial conditions $x_0 \in \mathcal{CS}^{12}$ and disturbance realizations, the safe policy~\eqref{eq:theSafePolicy} steered the system to the set $\mathcal{O}$, as shown in Figure~\ref{fig:safePolicy}.

\begin{figure}[h!]
\centering
\includegraphics[width=\columnwidth]{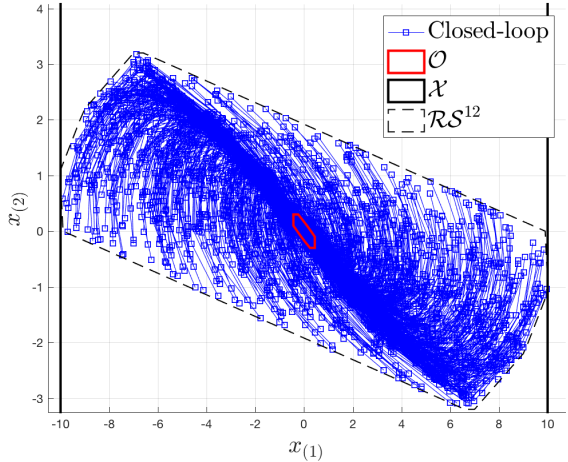}
\caption{Closed-loop trajectories for different disturbance realizations and initial conditions.}
\label{fig:safePolicy}
\end{figure}

We compare the performance of the uncertain system \eqref{eq:sys} in closed-loop with the safe policy~\eqref{eq:theSafePolicy} and the LMPC~\eqref{eq:FTOCP} and~\eqref{eq:MPCpolicy}. In particular, we simulated both the closed-loop system \eqref{eq:sys} and \eqref{eq:theSafePolicy} and the closed-loop system \eqref{eq:sys} and~\eqref{eq:MPCpolicy} for the same random initial condition $x_0 \in \mathcal{CS}^{12}$ and disturbance realization. As reported in Table~\ref{tab:comparison}, on average it takes $\sim 5$ms to evaluate the safe policy~\eqref{eq:theSafePolicy} and $\sim 4.6$s to evaluate the LMPC policy~\eqref{eq:MPCpolicy}. This result is expected as the safe policy~\eqref{eq:theSafePolicy} is evaluated solving a LP and the LMPC policy~\eqref{eq:MPCpolicy} solving $N+1=4$ QPs. On the other hand, it is interesting to notice that the closed-loop cost associated with the safe policy~\eqref{eq:theSafePolicy} is on average $\sim 3\%$ higher than the cost associated with the LMPC policy~\eqref{eq:MPCpolicy}, as shown in Table~\ref{tab:comparison}. This result suggests that, in applications where the computational power is not always available, one can first use the proposed LMPC to iteratively construct a large robust convex safe set $\mathcal{CS}^j$ and robust $Q$-function $Q^j$. Afterwards, these quantities can be leveraged to synthesize a safe control policy, which at the cost of slickly worse performance is able to reduce the computational burden.

\section{Conclusions}\label{sec:conclusions}
In this work we proposed a robust Learning Model Predictive Controller (LMPC) for linear systems subject to bounded additive uncertainty. 
At each execution of the control task, we store both the closed-loop data and the optimal predicted policy of the LMPC. First, we show how the stored data can be combined with the optimal policy from the LMPC to construct a safe set and an approximation to the value function. Afterwards, we design a safe policy which may be used to complete the task from any state in the safe set. Finally, the safe set, the value function approximation, and the safe policy are used in the control design, which guarantees input to state stability, robust constraint satisfaction, and performance bounds for the certainty equivalent closed-loop system. The effectiveness of the proposed LMPC is tested on a numerical example.

\section{Acknowledgment}
The authors would like to thank Monimoy Bujarbaruah and Siddharth Nair for helpful discussions and
reviews. Some of the research described in this review was funded by the Hyundai Center of Excellence
at the University of California, Berkeley. This work was also sponsored by the Office of Naval
Research. The views and conclusions contained herein are those of the authors and should not be
interpreted as necessarily representing the official policies or endorsements, either expressed or
implied, of the Office of Naval Research or the US government.

\section{Appendix}\label{sec:appendixProof}
\subsection{Proof of Proposition~\ref{prop:invariance}}
\begin{proof}
The proof follows from linearity of system~\eqref{eq:sys} and convexity of the constraint set~\eqref{eq:constraintSet}. \\
By Assumption~\ref{ass:policyAss}, for all vertices $v_{k|t}^{j,i}$ and associated control action $\pi^{j,*}_{k|t}( v_{k|t}^{j,i} )$ collected in the columns of the matrices ${\bf{X}}^j$ in~\eqref{eq:xMatrix} and  ${\bf{U}}^j$ in~\eqref{eq:uMatrix}, we have that
\begin{equation*}
    Av_{k|t}^{j,i} + B\pi^{j,*}_{k|t}( v_{k|t}^{j,i} ) + w \in \mathcal{CS}^j \subseteq \mathcal{X}, ~\forall w \in \mathcal{W}.
\end{equation*}
The above equation implies that $\forall x \in \mathcal{CS}^j$ and $\forall {\bm{\lambda}}^j \in \Lambda(x)$
\begin{equation*}
    A {\bf{X}}^j {\bm{\lambda}}^j + B {\bf{U}}^j {\bm{\lambda}}^j+w \in \mathcal{CS}^j, \forall w \in \mathcal{W}.
\end{equation*}
By definition $\forall \kappa^j(\cdot) \in \mathcal{SP}^j$ there exists ${\bm{\lambda}}^j \in \Lambda(x)$ such that $x = {\bf{X}}^j{\bm{\lambda}}^j$ and $\kappa^j(x) = {\bf{U}}^j {\bm{\lambda}}^j$. Consequently, from the above equation we have that $\forall x \in \mathcal{CS}^j$ and $\forall \kappa^j(\cdot) \in \mathcal{SP}^j$
\begin{equation*}
    A x + B \kappa^j(x) + w  \in \mathcal{CS}^j, \forall w \in \mathcal{W}.
\end{equation*}
\end{proof}

\subsection{Proof of Proposition~\ref{prop:QfunDecrease}}
\begin{proof}
Recall that we initialized $\mathcal{CS}^{0}=\mathcal{O}$ and $Q^{0}(x)=0~\forall x \in \mathcal{CS}^{0}$, then we trivially have that
\begin{equation*}
    Q^{0}(x) \geq h(x,  \kappa^{0,*}(x)) + Q^{0}(Ax+B \kappa^{0,*}(x)), \forall x \in \mathcal{CS}^{0}.
\end{equation*}
Next, we show that $\forall j \geq 1$ and $\forall x \in \mathcal{CS}^{j}$
\begin{equation*}
    Q^{j}(x) \geq h(x,  \kappa^{j,*}(x)) + Q^{j}(Ax+B \kappa^{j,*}(x)).
\end{equation*}
Given $ x \in \mathcal{CS}^j$, we have that
\begin{equation}\label{eq:defQProof}
    Q^j(x) = ({\bf{J}}^{j})^\top {\bm \lambda}^{*,j},
\end{equation}

Now notice that by definitions~\eqref{eq:J_definition1}--\eqref{eq:J_definition2} each element of ${\bf{J}}^{j}$ can be written as
\begin{equation}\label{eq:Qwritten_1}
\begin{aligned}
    J_{k|t}^{j}\big(v_{k|t}^{j,i}  \big) &= h\big(v_{k|t}^{j,i} ,  \pi_{k|t}^{j,*} (v_{k|t}^{j,i} ) \big) \\
    & \quad \quad \quad \quad \quad \quad +\sum_{r= 1}^{l^{k+1}} \gamma^{r,*} J^j_{k+1|t}\big(v_{k+1|t}^{j,r}\big)  \\
    &\geq h\big(v_{k|t}^{j,i} ,  \pi_{k|t}^{j,*} (v_{k|t}^{j,i} ) \big)  \\
    & \quad \quad \quad \quad + Q^j \big(A v_{k|t}^{j,i} +  B \pi_{k|t}^{j,*} (v_{k|t}^{j,i}) \big).
\end{aligned}
\end{equation}
where the above inequality holds as 
\begin{equation*}
    \sum_{r=1}^{l^{k+1}} \gamma^{r,*}  v_{k+1|t}^{j,r} = A v_{k|t}^{j,i} + B  \pi_{k|t}^{j,*} (v_{k|t}^{j,i})
\end{equation*} 
and for $x^+ = A v_{k|t}^{j,i} +  B \pi_{k|t}^{j,*} (v_{k|t}^{j,i})$
\begin{equation*}
Q^j \big( x^+ \big) = \min_{\boldsymbol{\lambda}^j \in \Lambda(x^+)} ({\bf{J}}^{j})^\top {\bm \lambda}^{j}.
\end{equation*}
Convexity of $h(\cdot, \cdot), Q^j(\cdot)$ and Equation~\eqref{eq:Qwritten_1} imply that
\begin{equation}
\begin{aligned}
     & Q^j(x)  = {\bf{J}}^j {\bm \lambda}^{*,j} \\ 
     & \geq h\Bigg(\begin{bmatrix} v_{0|0}^{0,1} \\ 
                                           \vdots \\ 
                                  v_{k|t}^{j,i} \\ 
                                  \vdots \\
                                  v_{t+T^j\text{-}1|T^j}^{l^{N\text{-}1},j} \end{bmatrix}  {{ { \bm \lambda}}}^{j,*} , 
      \begin{bmatrix} \pi_{0|0}^{0,*}(v_{0|0}^{0,1}) \\ 
                                                \vdots \\ 
                    \pi_{k|t}^{j,*}(v_{k|t}^{j,i}) \\ 
                                                \vdots \\
                    \pi_{t+T^j\text{-}1|T^j}^{j,*}(v_{t+T^j\text{-}1|T^j}^{j,l^{N\text{-}1}}) \end{bmatrix}  {{ { \bm \lambda}}}^{j,*}\Bigg)  \\
     & \quad \quad + \begin{bmatrix} Q^j\big(A v_{0|0}^{0,1} + B \pi_{0|0}^{0,*} (v_{0|0}^{0,1}) \big) \\ 
                                                                                                \vdots \\ 
                                        Q^j(A v_{k|t}^{j,i} + B \pi_{k|t}^{j,*} (v_{k|t}^{j,i}) \big) \\ 
                                        \vdots \\  
                                        Q^j\big(A v_{t+T^j\text{-}1|T^j}^{j,l^{N\text{-}1}} + B \pi_{t+T^j\text{-}1|T^j}^{j,*} ( v_{t+T^j\text{-}1|T^j}^{j,l^{N\text{-}1}} ) \big) \end{bmatrix} {\bm \lambda}^{*,j} \\
     & ~~ \geq h\big({\bf{X}}^j  {{ { \bm \lambda}}}^{j,*} , {\bf{U}}^j  {{ { \bm \lambda}}}^{j,*}\big)  + Q^j ( {\bf{X}}^j \tilde {\bm \lambda}^{j,*}),
\end{aligned}
\end{equation}
for some $\tilde {\bm \lambda}^{j,*}$ such that
\begin{equation*}
    {\bf{X}}^j \tilde {\bm \lambda}^{j,*} = \begin{bmatrix} A v_{0|0}^{0,1} + B \pi_{0|0}^{0,*} \big(v_{0|0}^{0,1} \big) \\ \vdots \\ A v_{k|t}^{j,i} + B \pi_{k|t}^{j,*} \big(v_{k|t}^{j,i} \big) \\ \vdots \\ A v_{t+T^j\text{-}1|T^j}^{j,l^{N\text{-}1}} + B \pi_{t+T^j\text{-}1|T^j}^{j,*} \big(v_{t+T^j\text{-}1|T^j}^{j,l^{N\text{-}1}} \big) \end{bmatrix} {\bm \lambda}^{j,*}.
\end{equation*}
The above equation implies that
\begin{equation*}
\begin{aligned}
    Q^j(x) & \geq h\big({\bf{X}}^j  {{\bm \lambda}}^{j,*} , {\bf{U}}^j  {{\bm \lambda}}^{j,*}\big)  + Q^j ( {\bf{X}}^j \tilde {\bm \lambda}^{j,*}) \\
    & \geq h({\bf{X}}^j  {\bm \lambda}^{j,*} , {\bf{U}}^j {\bm \lambda}^{j,*} ) + Q^j( A{\bf{X}}^j  {\bm \lambda}^{j,*} + B {\bf{U}}^j {\bm \lambda}^{j,*} )
\end{aligned}
\end{equation*}
Finally, we notice that by definition~\eqref{eq:theSafePolicy} $\kappa^{j,*}( x )={\bf{U}}^j {\bm \lambda}^{j,*}$, therefore the above equation can be rewritten as 
\begin{equation*}
    Q^j(x) \geq h(x , \kappa^{j,*}(x) ) + Q^{j}(Ax + B \kappa^{j,*}(x)).
\end{equation*}
\end{proof}

\subsection{Proof of Proposition~\ref{prop:safePolicyISS}}
\begin{proof} First we show that the set $\mathcal{O}$ is robust positive invariant for the closed-loop system~\eqref{eq:sys} and \eqref{eq:theSafePolicy}. Assume that at time $t$ of iteration $j$ the state $x_t^j \in \mathcal{O}$, by definitions~\eqref{eq:xMatrix}, \eqref{eq:J_definition1}, \eqref{eq:Qfun}, and \eqref{eq:theSafePolicy} we have that $Q^j(x_t^j)=0$ and $\kappa^{j,*}(x_{t}^{j})=Kx_t^j$. Therefore, by Assumption~\ref{ass:O_inf} we have that closed-loop system $x_{t+1}^j = Ax_t^j + B\kappa^{j,*}(x_{t}^{j}) +w_t^j \in \mathcal{O}, \forall w_t^j \in \mathcal{W}$ and that $\mathcal{O}$ is robust positive invariant for the closed-loop system~\eqref{eq:sys} and \eqref{eq:theSafePolicy}. \\
Continuity of $Q^j( \cdot )$, Assumption~\ref{ass:cost} and \eqref{eq:FTOPC_Cost} imply the existence of $\alpha_1, \alpha_2 \in \mathcal{K}_\infty$ such that  $\forall x_t \in \mathcal{CS}^j$
\begin{equation}
    \alpha_1(|x_t|_\mathcal{O}) \leq h(x_t, 0) \leq Q^j( x_t ) \leq \alpha_2(|x_t|_\mathcal{O}).
\end{equation}
We recall that $Q^j( \cdot )$ is Lipschitz continuous as it is the solution to a parametric LP \cite{BorrelliBemporadMorari_book}.
Therefore, from Proposition~\ref{prop:QfunDecrease} and Lipschitz continuity of $Q^j( \cdot )$ for a Lipschitz constant $L$, we have that $\forall x_t \in \mathcal{CS}^j$
\begin{equation*}
\begin{aligned}
     Q^j&( Ax_t+B\kappa^{j,*}(x_t)+w_t )-Q^j( x_t ) \\
     &=Q^j( Ax_t+B\kappa^{j,*}(x_t)+w_t )-Q^j( Ax_t+B\kappa^{j,*}(x_t))\\&+Q^j( Ax_t+B\kappa^{j,*}(x_t))-Q^j( x_t )\\
     &~\leq L||w_t||_2+Q^j( Ax_t+B\kappa^{j,*}(x_t))-Q^j( x_t )\\
     &\leq L||w_t||_2 - h(x_t, \kappa^{j,*}(x_t)).
\end{aligned}
\end{equation*}
Therefore $Q^j(\cdot)$ is a ISS-Lyapunov function according with Definition~\ref{def:iss} and by Proposition~\ref{prop:ISS} the closed-loop system~\eqref{eq:sys} and~\eqref{eq:theSafePolicy} is ISS for the robust positive invariant set $\mathcal{O}$. 
\end{proof}

\subsection{Proof of Proposition~\ref{prop:costBound}}
\begin{proof} By assumption $x_t^j \in \mathcal{C}^j$, therefore by Theorem~\ref{th:recFeas} Problem~\eqref{eq:FTOCP} is feasible at all times $t\geq0$ for some $N_t^j \in \{0,\ldots, N\}$. Let $\boldsymbol{\pi}_t^{j,*}(\cdot)$ be the optimal $N$-steps policy~\eqref{eq:optNstepPolicy} and  $J_{t\rightarrow t+N}^{\scalebox{0.4}{LMPC},j}(x_t^j)$ the optimal cost. At time $t$, let
\begin{equation*}
    [\bar x_{t|t}^{j,*}, \ldots, \bar x_{t+N|t}^{j,*} ]
\end{equation*}
be the optimal trajectory of the certainty equivalent system associated with the optimal policy $\boldsymbol{\pi}_t^{j,*}(\cdot)$. Then, the LMPC cost at time $t$ can be written as 
\begin{equation}\label{eq:cost_t}
\begin{aligned}
    J^{\scalebox{0.5}{LMPC},j}_{t \rightarrow t+N}( x_t^j ) &= \sum_{k=t}^{t+N-1} h(\bar x_{k|t}^{j,*}, \pi_{k|t}^{j,*}(\bar x_{k|t}^{j,*})) + Q^{j-1}(\bar x_{t+N|t}^{j,*})\\
    &= h(\bar x_{t|t}^{j,*}, \pi_{t|t}^{j,*}(\bar x_{k|t}^{j,*})) + p(\bar x_{t+1|t}^{j,*})
\end{aligned}
\end{equation}
where the function $p(\bar x_{t+1|t}^{j,*})$, which represents the total cost from time $t+1$ to time $t+N$ for the optimal policy $\boldsymbol{\pi}_t^{j,*}(\cdot)$, is Lipschitz as it is composed of summation and composition of Lipschitz functions. Now we notice that by feasibility of \eqref{eq:candidateSolution} the cost of the LMPC at time $t+1$ satisfies the following inequality
\begin{equation}\label{eq:cost_t+1}
\begin{aligned}
    &J^{\scalebox{0.5}{LMPC},j}_{t+1 \rightarrow t+1+N}(f(x_t,w_t)) \leq \sum_{k=t+1}^{t+N-1} h\big(\bar x_{k|t+1}^j ,  \pi_{k|t}^{*,j}(\bar x_{k|t+1}^j) \big) \\
    &~~+ h\big(\bar x_{t+N|t+1}^j, \kappa^{j-1,*} (\bar x_{t+N|t+1}^j )\big)\\
    &~~+ Q^{j-1}\big(A \bar x_{t+N|t+1}^j+B \kappa^{j-1,*}(\bar x_{t+N|t+1}^j)\big)\\
    &~~\leq \sum_{k=t+1}^{t+N-1} h\big(\bar x_{k|t+1}^j ,  \pi_{k|t}^{*,j}(\bar x_{k|t+1}^j) \big)+ Q^{j-1}\big( \bar x_{t+N|t+1}^j \big)\\
    &~~= p(\bar x_{t+1|t+1}^j),
\end{aligned}
\end{equation}
where the last inequality follows from Proposition~\ref{prop:QfunDecrease}, the feasible nominal trajectory is given by
\begin{equation*}
\begin{aligned}
    \bar x_{k|t+1}^j = A^{k-t-1}(Ax_t+&B\pi_{t|t}^{*,j}(x_t) + w_t) \\
    &+ \sum_{i=t+1}^{k-1}A^{k-1-i}B\pi_{i|t}^{*,j}(\bar x_{i|t+1}^j)
\end{aligned}
\end{equation*} 
for $k=\{t+2, \ldots, t+N \}$, and the initial state
\begin{equation}\label{eq:equalityProp}
    \bar x_{t+1|t+1}^j = Ax_t^j+B\pi_{t|t}^{*,j}(x_t^j) + w_t^j = \bar x_{t+1|t}^{j,*} + w_t^j,
\end{equation}
$\forall w_t^j \in \mathcal{W}$. Therefore, from the $L$-Lipschitz continuity of $p(\cdot)$, \eqref{eq:cost_t}, \eqref{eq:cost_t+1} and \eqref{eq:equalityProp} we have that
\begin{equation}
\begin{aligned}
    & J^{\scalebox{0.5}{LMPC},j}_{t+1 \rightarrow t+1+N}(f(x_t^j,w_t^j)) - J^{\scalebox{0.5}{LMPC},j}_{t \rightarrow t+N}(x_t^j)\\
    &\quad \quad \quad = p(\bar x_{t+1|t}^{j,*} + w_t)-p(\bar x_{t+1|t}^{j,*}) -h(\bar x_{t|t}^{*,j},\pi_{t|t}^{*,j}( x_t^j))\\
    &\quad \quad \quad \quad \quad \leq -h(\bar x_{t|t}^{*,j},\pi_{t|t}^{*,j}(x_t^j)) + L ||w_t^j||\\
    &\quad \quad \quad \quad \quad \quad \leq -h(\bar x_{t|t}^{*,j},0) + L ||w_t^j||\\
    &\quad \quad \quad \quad \quad \quad \quad \leq -\alpha_x^l( x_{t}^{j}) + L ||w_t^j||,
\end{aligned}
\end{equation}
for all $w_t^j \in \mathcal{W}$, where the last inequality holds by Assumption~\ref{ass:cost}.
\end{proof}




\bibliographystyle{IEEEtran}
\bibliography{IEEEabrv,mybibfile}

\end{document}